\documentclass[11pt,reqno]{amsart}
\usepackage[utf8]{inputenc}
\usepackage{amssymb,amsthm,mathrsfs,mathtools,color,dsfont,enumerate}
\usepackage{a4wide}
\usepackage{graphics}
\usepackage{epsfig}
\usepackage{caption}
\usepackage{subcaption}
\usepackage{xcolor}    
\usepackage[colorlinks=true]{hyperref}
\hypersetup{,urlcolor=blue, citecolor=red,linkcolor=blue}
\usepackage[]{diagbox}
\usepackage{multirow}
\usepackage{bm} %for bold in mathmode
\usepackage{floatrow} % Table float box with bottom caption, box width adjusted to content
\usepackage{algorithm,algorithmic}

\newtheorem{theorem}{Theorem}[section]
\newtheorem{prop}[theorem]{Proposition}
\newtheorem{lemma}[theorem]{Lemma}
\newtheorem{corollary}[theorem]{Corollary}

\newtheorem{remark}[theorem]{Remark}

\newcommand{\R}{{\mathbb R}}
\newcommand{\N}{{\mathbb N}}
\newcommand{\E}{{\mathbb E}}

\newcommand{\cPh}{\hat{\mathcal P}}
\newcommand{\cI}{{\mathcal I}}

\renewcommand{\L}{{\mathbb L}}
\newcommand{\pol}{\textup{pol}}
\newcommand{\CP}[1]{{\mathcal C}^{#1}_{\pol}(\R\times \R_+)}
\newcommand{\CPL}[2]{{\mathcal C}^{#1,#2}_{\pol}(\R\times \R_+)}

\def\bbar{\overline}
\def\brho{\bbar \rho}

\def\hX{\hat X}
\def\hY{\hat Y}

\def\tK{\tilde K}

\def\tX{\tilde X}
\def\tY{\tilde Y}

\def\mcC{{\mathcal C}}
\def\mcL{{\mathcal L}}
\def\mcM{{\mathcal M}}
\def\mcN{{\mathcal N}}
\def\bfy{{\mathbf y}}

\def\bff{{\mathbf f}}
\def\tbfY{\tilde{\mathbf Y}}
\def\a{\alpha}
\def\b{\beta}

\def\<{\langle}
\def\>{\rangle}

\def \pol{{\mathbf{pol}}}

\def\red{}%\textcolor{red}}
\makeatletter
 \@addtoreset{equation}{section}
 \makeatother
 
\allowdisplaybreaks
\usepackage{appendix}
\title[High order approximations for the Heston model]{High order approximations and simulation schemes for the log-Heston process}
\date{\today}
\author{Aurélien Alfonsi and Edoardo Lombardo}
\address{Aurélien Alfonsi, CERMICS, Ecole des Ponts, Marne-la-Vall\'ee, France. MathRisk, Inria, Paris, France.}
\email{aurelien.alfonsi@enpc.fr}
\address{Edoardo Lombardo, CERMICS, Ecole des Ponts, Marne-la-Vall\'ee, France. MathRisk, Inria, Paris, France. Universit\`a degli Studi di Roma Tor Vergata, Rome, Italy. }
\email{edoardolombardo92@gmail.com}
\thanks{This work benefited from the support of the ``chaire Risques financiers'', Fondation du Risque.  Edoardo Lombardo is partially supported by the MIUR Excellence Department Project Math@TOV awarded to the Department of Mathematics, University of Rome Tor Vergata. We are grateful to the anonymous referees for their feedback on the first version of this paper.}



\begin{document}

\begin{abstract}
  We present weak approximations schemes of any order for the Heston model that are obtained by using the method developed by Alfonsi and Bally (2021). This method consists in combining approximation schemes calculated on different random grids to increase the order of convergence. We apply this method with either the Ninomiya-Victoir scheme (2008) or a second-order scheme that samples exactly the volatility component, and we show rigorously that we can achieve then any order of convergence. We give numerical illustrations on financial examples that validate the theoretical order of convergence. \red{We also present promising numerical results for the multifactor/rough Heston model and hint at applications to other models, including the Bates model and the double Heston model. }
\end{abstract}
\maketitle
%%%%%%%%%%%%%%%%%%%%%%%%%%%%%%%%%%%%%%%%%%%%%%%%%%%%%%%%%%%
%%%%%%%%%%%%%%%%%%%%%%%%%%%%%%%%%%%%%%%%%%%%%%%%%%%%%%%%%%%

\section{Introduction}

The Heston model~\cite{Heston} is one of the most popular model in mathematical finance. It describes the dynamics of an asset and its instantaneous volatility by the following stochastic differential equations:
\begin{equation}\label{Heston_SDEs}
  \begin{cases}
    dS^{s,y}_t &= rS^{s,y}_t dt + \sqrt{Y^y_t}S^{s,y}_t (\rho dW_t + \sqrt{1-\rho^2} dB_t), \ S^{s,y}_0=s >0,\\
    dY^y_t &= (a-bY^y_t) dt +\sigma \sqrt{Y^y_t} dW_t, \ Y^y_0=y\ge 0,
  \end{cases}
\end{equation}
where $W$ and $B$ are two independent Brownian motions, $a\ge 0$, $b\in \R$, $\sigma>0$ and $\rho \in [-1,1]$. For the financial application, it is typically assumed in addition that $b>0$ so that the volatility mean reverts towards $a/b$, but this is not needed in the present paper. 

The goal of the paper is to propose high order weak approximation for this model and to prove their convergence. Let us recall first that exact simulation methods have been proposed for the Heston model by Broadie and Kaya~\cite{BrKa} and then by Glasserman and Kim~\cite{GlKi}. However, these methods usually require more computation time than simulation schemes. Besides, when considering variants or extensions of the Heston model, it is not clear how to simulate them exactly while approximation schemes can more simply be reused or adapted. There exists in the literature many approximation schemes of the Heston model, we mention here the works of  Andersen~\cite{Andersen}, Lord et al.~\cite{LKVD}, Ninomiya and Victoir~\cite{NV} and Alfonsi~\cite{AA_MCOM}. Few of them study from a theoretical point of view the weak convergence of these schemes. While~\cite{AA_MCOM} focuses on the volatility component, Altmayer and Neuenkirch~\cite{AlNe} proposes up to our knowledge the first analysis of the weak error for the whole Heston model. They essentially obtain for a given Euler/Milstein scheme a weak convergence rate of~$1$ under the restriction $\sigma^2<a$ on the parameter.

Like~\cite{AlNe}, we will work with the log-Heston model that solves the following SDE
\begin{equation}\label{log-Heston_SDEs}
  \begin{cases}
    dX^{x,y}_t &= (r-\frac{1}{2}Y^y_t) dt + \sqrt{Y^y_t} (\rho dW_t + \sqrt{1-\rho^2} dB_t), \ X^{x,y}_0=x=\log(s) \in \R,\\
    dY^y_t &= (a-bY^y_t) dt +\sigma \sqrt{Y^y_t} dW_t, \ Y^y_0=y.
  \end{cases}
\end{equation}
This log transformation of the asset price is standard to carry mathematical analyses: it allows to get an SDE with bounded moments since its coefficients have at most a linear growth. Our goal is to propose approximations of any order of the semigroup $P_Tf(x,y)=\E[f(X^{x,y}_T,Y^y_T)]$, where $f:\R \times \R_+ \to \R$ is a sufficiently smooth function such that $|f(x,y)|\le C(1+|x|^L+y^L)$ for some $L \in \N$. More precisely, we want to apply the recent method proposed by Alfonsi and Bally~\cite{AB} that allows to boost the convergence of an approximation scheme by using random time grids. We consider here either the Ninomiya-Victoir scheme for $\sigma^2\le 4a$ or a scheme that simulate exactly $Y$ for any $\sigma>0$.
In a previous work~\cite{AL}, we have applied the method of~\cite{AB} to the only Cox-Ingersoll-Ross component~$Y$ and we want to extend our result to the full log-Heston model. 
 One difficulty with respect to the general framework developed in~\cite{AB} is to deal with the singularity of the diffusion coefficient. In particular, we need some analytical results on the Cauchy problem associated to the log-Heston model that have been obtained recently by Briani et al.~\cite{BrCaTe}. Our main theoretical result (Theorem~\ref{main_theorem}) states that we obtain, for any $\nu\ge 1$, semigroup approximations of order $2\nu$ by using the mentioned scheme with the boosting method of~\cite{AB}.

 The paper is structured as follows. Section~\ref{Sec_main} presents the precise framework and in particular the functional spaces that we consider carrying our analysis. It introduces the approximation schemes and briefly presents the boosting method using random grids proposed in~\cite{AB}. The main result of the paper is then stated precisely.  Section~\ref{Sec_proof} is dedicated to the proof of the main theorem. Last, Section~\ref{Sec_num} explains how to implement our approximations and illustrates their convergence on numerical experiments motivated by the financial application. As an opening for future research, we show that our approximations can be used for the multifactor Heston model\footnote{We recall that the multifactor Heston model proposed by Abi Jaber and El Euch~\cite{AEE19} is an extension of the Heston model that is a good proxy of the rough Heston model introduced by El Euch and Rosenbaum~\cite{EER19}.} under some parameter restrictions and give very promising convergence results. \red{We also indicate a wide class of models that includes the Bates model and the double Heston model to which our approximations can be applied.}

\section{Main results}\label{Sec_main}

We start by introducing some functional spaces that are used through the paper.
For $k\in\N$, we denote by $\mcC^{k}(\R\times\R_+)$ the space of continuous functions $f:\R \times \R_+ \to \R$ such that the partial derivatives $\partial^\a_x \partial^\b_y  f(x,y)$ exist and are continuous with respect to $(x,y)$ for all $(\a,\b)\in\N^2$ such that $\a+2\b\le k$. We then define for $L\in \N$,
\begin{equation}\label{def_fL}
   {\bf f}_L(x,y)=  (1+x^{2L}+y^{2L}), \ x\in \R,y\in \R_+,
\end{equation}
and introduce 
\begin{multline}\label{def_CpolkL}
  \CPL{k}{L}=\{ f\in \mcC^{k}(\R\times\R_+) \mid \exists C>0 \text{ such that } \forall(\a,\b)\in\N^2, \a+2\b\le k,\\
  |\partial^\a_x \partial^\b_y  f(x,y)| \le C {\bf f}_L(x,y)   \},
\end{multline}
the space of continuously differentiable functions up to order~$k$ with derivatives with polynomial growth of order $2L$. Note that we assume twice less differentiability on the $y$ component.  
Furthermore, we set 
$$\CP{k}= \cup_{L\in \N} \CPL{k}{L} \text{ and } \CP{\infty}= \cap_{k\in \N} \CP{k}.$$
Last, we endow $\CPL{k}{L}$ with the following norm:
\begin{equation}\label{def_NormCpolKL}
  \|f\|_{k,L}=\sum_{\alpha+2\beta \le k} \sup_{(x,y) \in \R \times \R_+} \frac{|\partial^\alpha_x\partial^\beta_y f(x,y)|}{\bff_L(x,y)}.
\end{equation}

\subsection{Second order schemes for the log-Heston process}\label{subsec_2nd}

Having in mind~\cite[Theorem 2.3.8]{AA_book}, there are three properties to check to get a second-order scheme for the weak error:
\begin{enumerate}[(a)]
  \item polynomial estimates for the derivatives of the solution of the Cauchy problem,
  \item uniformly bounded moments of the approximation scheme, 
  \item a potential second order scheme, which roughly means that we have a family of random variables $(\hat{X}^{x,y}_t,\hat{Y}^y_t)_{t\ge 0}$ such that $|\E[f(\hat{X}^{x,y}_t,\hat{Y}^y_t)]-f(x,y)-t \mathcal{L} f(x,y)- \frac{t^2}2\mathcal{L}^2 f(x,y)|=_{t\to 0} O(t^3)$. 
\end{enumerate}

Let us precise this in our context. We consider a time horizon $T>0$ and a time step $h=T/n$, with $n\in \N^*$. We note $(\hat{X}^{x,y}_h,\hat{Y}^y_h)$ an approximation scheme for the SDE~\eqref{log-Heston_SDEs} starting from $(x,y)$ with time-step $h$, and $$\hat{P}_h f(x,y)=\E[f(\hat{X}^{x,y}_h,\hat{Y}^y_h)]$$ the associated semigroup approximation. The weak error analysis proposed by Talay and Tubaro~\cite{TT} consists in writing
\begin{equation}\label{eq_TT}
  \hat{P}^{[n]}_h-P_T=\hat{P}^{[n]}_h-P_h^{[n]}=\sum_{i=0}^{n-1}\hat{P}_h^{[n-(i+1)]}(\hat{P}_{h}-P_h)P_h^{[i]}=\sum_{i=0}^{n-1}\hat{P}_h^{[n-(i+1)]}(\hat{P}_{h}-P_h)P_{ih},
\end{equation}
where $\hat{P}_h^{[0]}=Id$ and $\hat{P}_h^{[i]}=\hat{P}_h^{[i-1]}\hat{P}_h$ for $i\ge 1$, and $P_h^{[i]}=P_{ih}$ by the semigroup property.   
Let us assume that the three properties hold
\begin{enumerate}[(a)]
\item $\forall k,L \in \N, \ \exists C \in \R_+, \ \forall i\in \{0,\dots,n\}, \  \|P_{ih} f\|_{k,L} \le C \|f\|_{k,L}$,
\item $\forall L \in \N, \exists C_L \in \R_+, \ \max_{0\le i\le n} \hat{P}_h^{[i]} \bff_L (x,y) \le C_L \bff_L(x,y)$,
\item $\|\hat{P}_h f-P_hf\|_{0,L+3}\le Ch^3 \|f\|_{12,L}.$
\end{enumerate}
Then, for $f \in \CPL{12}{L}$, we have for each $i \in \{0,\dots,N-1\},$
$$\|(\hat{P}_{h}-P_h)P_{ih}f \|_{0,L+3}\le Ch^3\|P_{ih}f \|_{12,L}\le C^2 h^3\|f \|_{12,L}, $$
by using the first and third properties. Then, we use that $$|(\hat{P}_{h}-P_h)P_{ih}f(x,y)|\le  C^2\|f \|_{12,L} h^3 \bff_{L+3}(x,y),$$
together with the second property to get that $|\hat{P}^{[n-(i+1)]}(\hat{P}_{h}-P_h)P_{ih}f(x,y)|\le C_LC^2 h^3\bff_L(x,y)$. This bound is uniform with respect to~$i$, and we get
\begin{equation}
  |\hat{P}^{[n]}_hf(x,y)-P_Tf(x,y)|\le C_LC^2 T \bff_{L+3}(x,y) \times \left(\frac{T}{n}\right)^2,
\end{equation}
since $h=T/n$.

Before concluding this paragraph, we comment briefly how to get the three properties (a--c). For the log-Heston SDE, the Cauchy problem has been studied by Briani et al.~\cite{BrCaTe} and their analysis allow to get~(a). Their result is reported in Proposition~\ref{prop-rep-logHeston-estim}. Property (b) can generally be checked by simple but sometimes tedious calculation. Property (c) is the crucial one and can be obtained by using splitting technique, see~\cite[Paragraph 2.3.2]{AA_book}. We check this property in Corollary~\ref{cor_H1} for the schemes presented in this paper.

\subsection{From the second order scheme to higher orders by using random grids}

We continue the analysis and present, in our framework, the method developed by Alfonsi and Bally~\cite{AB} to get approximations of any orders by using random grids. For $l \in \N^*$, let us define the time step
$h_l=\frac{T}{n^l}$. We set $Q_l=\hat{P}_{h_l}$ the operator obtained by using the approximation scheme with the time step $h_l$. The principle is to iterate the identity~\eqref{eq_TT}. Namely, we get from~\eqref{eq_TT}
$$ \hat{P}_{h_1}^{[i]}-P_{ih_1}=\sum_{i_1=0}^{i-1} \hat{P}^{[i-(i_1+1)]}_{h_1}(\hat{P}_{h_1}-P_{h_1})P^{[i_1]}_{h_1} $$
and 
$$ \hat{P}_{h_2}^{[n]}-P_{h_1}=\hat{P}_{h_2}^{[n]}-P_{h_2}^{[n]}=\sum_{j=0}^{n-1}\hat{P}^{[n-(j+1)]}_{h_2}(\hat{P}_{h_2}-P_{h_2})P^{[j]}_{h_2}.   $$
Plugging these two identities successively in~\eqref{eq_TT}, we obtain
\begin{equation}\label{eq_boost2}\hat{P}^{[n]}_{h_1}-P_T=\sum_{i=0}^{n-1}\hat{P}_{h_1}^{[n-(i+1)]}(\hat{P}_{h_1}-\hat{P}_{h_2}^{[n]})\hat{P}^{[i]}_{h_1}+R,\end{equation}
with the remainder given by
\begin{align*}
  R=&\sum_{i=0}^{n-1}\hat{P}_{h_1}^{[n-(i+1)]} \left(\sum_{j=0}^{n-1}\hat{P}^{[n-(j+1)]}_{h_2}(\hat{P}_{h_2}-P_{h_2})P^{[j]}_{h_2} \right)\hat{P}^{[i]}_{h_1} \\&- \sum_{i=0}^{n-1}\hat{P}_{h_1}^{[n-(i+1)]} (\hat{P}_{h_1}-P_{h_1}) \sum_{i_1=0}^{i-1} \hat{P}^{[i-(i_1+1)]}_{h_1}(\hat{P}_{h_1}-P_{h_1})P^{[i_1]}_{h_1} .\end{align*}
Let us assume that we have the two following properties\footnote{We directly specify the method to our framework, and refer to~\cite{AB} or~\cite[Section 2]{AL} for a general presentation.}
\begin{align}\tag{$\bbar{H_1}$}\label{H1_bar}
   &\forall l,k,L\in\N,\exists C>0,\forall f\in \CPL{k+12}{L},\, \|(P_{h_l}-Q_l)f\|_{k,L+3} \leq C\|f\|_{k+12,L} h_l^{3},\\
&\tag{$\bbar{H_2}$} \label{H2_bar}
  \forall l,k,L\in \N,\exists C>0,\forall f\in \CPL{k}{L}, \, \max_{0\leq j\leq n^l}\|Q^{[j]}_l f\|_{k,L} + \sup_{t\leq T}\|P_t f\|_{k,L} \leq C\|f\|_{k,L}.
\end{align}
Then, we can upper bound the remainder for $f\in  \CPL{k+24}{L}$ by
$$ \|R f\|_{k,L+6}\le C^3n^2 \|f\|_{k+12,L+3} h_2^3+C^5 \frac{n(n-1)}2 \|f\|_{k+24,L} (h_1^3)^2 \le \tilde{C}\|f\|_{k+24,L} \left(\frac{T}{n}\right)^4,$$
where we have used twice~\eqref{H2_bar} and once~\eqref{H1_bar} for the first sum, and three times~\eqref{H2_bar} and twice~\eqref{H1_bar} for the second one. Therefore, we get from~\eqref{eq_boost2} that
\begin{equation}\label{def_boost2}
  \cPh^{2,n}:= \hat{P}^{[n]}_{h_1}+\sum_{i=0}^{n-1}\hat{P}_{h_1}^{[n-(i+1)]}(\hat{P}_{h_2}^{[n]}-\hat{P}_{h_1})\hat{P}^{[i]}_{h_1}
\end{equation}
is an approximation of order~$4$. Namely, we get
\begin{equation}\label{boost2_estimate} \forall f \in \CP{24}, \  \exists C>0,L\in \N, \   \|\cPh^{2,n} f -P_Tf\|_{0,L+6}\le C \|f\|_{24,L} \left(\frac{T}{n} \right)^4. \end{equation}
Let us note that $\hat{P}_{h_1}^{[n-(i+1)]}\hat{P}_{h_2}^{[n]}\hat{P}_{h_1}^{[i]}$ corresponds to the scheme on a time grid that is uniform, but uniformly refined on the $(i+1)$-th time step. This time grid has thus $2n$ time steps, and if one should calculate all the terms in the sum of~\eqref{def_boost2}, this would require a computational time in $O(n^2)$. Thus, the method would not be more efficient that using the second-order scheme with a time step $n^2$. This is why we use random grids and use a random variable $\kappa$ that is uniformly distributed on $\{0,\dots,n-1\}$. We have 
\begin{equation}\label{boost2_RG} \cPh^{2,n} = \hat{P}^{[n]}_{h_1}+n\E[\hat{P}_{h_1}^{[n-(\kappa+1)]}(\hat{P}_{h_2}^{[n]}-\hat{P}_{h_1})\hat{P}^{[\kappa]}_{h_1}].\end{equation}
Thus, for the correcting term, we consider a random time grid that is the obtained from the uniform time grid with time step $T/n$ by refining uniformly the $(\kappa+1)$-th time step with a time step $h_2=T/n^2$. 

We have presented here how $\cPh^{2,n}$ improves the convergence of $\cPh^{1,n}=\hat{P}^{[n]}_{h_1}$. Then, for $\nu \ge 2$, it is possible to define by induction approximations $\cPh^{\nu,n}$, such that 
\begin{equation}\label{def_Pnu}\forall f \in \CP{12 \nu}, \  \exists C>0,L\in \N, \   \|\cPh^{\nu,n} f -P_Tf\|_{0,L+3\nu}\le C \|f\|_{12\nu,L} \left(\frac{T}{n} \right)^{2 \nu}. \end{equation}
Unfortunately, the induction cannot be easily described and involves a tree structure. We refer to~\cite{AB} for the details and to~\cite[Eq. (2.8)]{AL} for an explicit expression of $\cPh^{3,n}$.

\subsection{A second-order scheme for the log-Heston model}

Before presenting the scheme, it is interesting to point similarities and  difference between the weak error analysis  of Subsection~\ref{subsec_2nd} and the present one to get higher order approximations. Property~\eqref{H1_bar} is a generalization of~(c), while~\eqref{H2_bar} is stronger than properties (a) and (b)\footnote{Note that $\bff_L\in \CPL{\infty}{L}$. We have, for $i \le n$, $\|P^{[i]}_{T/n} \bff_L\|_{0,L} =\|Q_1^{[i]}\bff_L\|_{0,L}\le C \|\bff_L\|_{0,L}$ by \eqref{H2_bar}, which gives (b).}. We now point an important difference between the two error analysis. In Equation~\eqref{eq_TT}, the difference between the semigroup and its approximation appears only once and there is no need to have regularity properties for the function $(\hat{P}_{h}-P_h)P_{ih}f$: only a polynomial bound is needed. In contrast, for the approximation $\cPh^{2,n}$ we need some regularity to iterate and upper bound the remainder. This difference has an important incidence in the case of the log-Heston process. It is proposed in~\cite{AA_MCOM} a second-order scheme for the log-Heston process for any $\sigma \ge 0$. When $\sigma^2\ge 4a$, this scheme relies for the Cox-Ingersoll-Ross (CIR) part on bounded random variables that match the first moments of the standard normal distribution. Unfortunately, these moment-matching variables prevent us to get~\eqref{H2_bar}: in a recent work on high order approximations for the CIR process, we have shown in~\cite[Theorem 5.16]{AL} that it was not possible to use these moment-matching variables together with our analysis in order to get~\eqref{H2_bar}. We do not repeat here the analysis that would be quite similar for the log-Heston model, and consider either the Ninomiya-Victoir scheme for $\sigma^2\le 4a$ or the exact CIR simulation for any $\sigma>0$. We now present this in detail.

We present in this subsection the approximations scheme that we will study in this paper. It is constructed by using the splitting technique. Let  $\brho=\sqrt{1-\rho^2}$, the  infinitesimal generator associated to the log-Heston SDE~\eqref{log-Heston_SDEs} is given by 
\begin{equation}\label{log_Heston2_diff-op}
  \mcL = \frac{y}{2}(\partial^2_x + 2\rho\sigma\partial_x\partial_y+\sigma^2\partial_y)+ (r-\frac{y}{2})\partial_x + (a-by)\partial_y.
\end{equation}
We split this operator as the sum $\mcL = \mcL_B+\mcL_W$ where \begin{equation}\label{def_LB}\mcL_B = \big((r-\frac{\rho a}{\sigma})+(\frac{\rho b}{\sigma}-\frac 12)y\big) \partial_x +\frac{y}{2}\brho^2\partial_x^2\end{equation}  is the infinitesimal generator of the SDE
\begin{equation}\label{LB_SDEs}
  \begin{cases}
    dX_t &= \big((r-\frac{\rho a}{\sigma})+(\frac{\rho b}{\sigma}-\frac 12)Y_t\big) dt+\brho\sqrt{Y_t} dB_t,\\
    dY_t &= 0,
  \end{cases}
\end{equation}
and  \begin{equation}\label{def_LW}\mcL_W =\frac{y}{2}(\rho^2\partial^2_x + 2\rho\sigma\partial_x\partial_y+\sigma^2\partial^2_y) + (a-by)(\frac{\rho}{\sigma}\partial_x+\partial_y)\end{equation} is the infinitesimal generator of 
\begin{equation}\label{LW_SDEs}
  \begin{cases}
    dX_t &= (\frac{\rho a}{\sigma}-\frac{\rho b}{\sigma}Y_t) dt +\rho \sqrt{Y_t} dW_t,\\
    dY_t &= (a-bY_t) dt +\sigma \sqrt{Y_t} dW_t.
  \end{cases}
\end{equation}
This splitting is slightly different from the one considered in~\cite[Subsection 4.2]{AA_MCOM}: one should remark that it is chosen in order to have in~\eqref{LW_SDEs} $dX_t = \frac{\rho}{\sigma}dY_t$. This is not particularly useful to get a second order scheme. However, it avoids introducing a third coordinate corresponding to the integrated CIR process, which is more parsimonious for the mathematical  analysis.

We now present two different second order schemes for the log-Heston process, for which we will be able to prove the effectiveness of the higher order approximations. The first one simply consists in sampling exactly each SDE and then using the scheme composition introduced by Strang, see e.g.~\cite[Corollary 2.3.14]{AA_book}. More precisely, let $P^B$ (resp. $P^W$) denote the semigroup associated to the SDE~\eqref{LB_SDEs} (resp.~\eqref{LW_SDEs}). We define
\begin{equation}\label{def_PEx}
  \hat{P}_t^{Ex}=P^B_{t/2}P^W_t P^B_{t/2}.
\end{equation}
Let us note that the exact scheme for~\eqref{LB_SDEs} is explicit and given by
\begin{equation}\label{varphiB}
  \varphi_B(t,x,y,N)=(x+(r-\rho a/\sigma)t -(1/2-\rho b/\sigma)yt+\brho\sqrt{ty}N,~y), \ \text{ with } N\sim\mcN(0,1).
\end{equation}
We  indeed have $P^B_tf(x,y)=\E[f(\varphi_B(t,x,y,N))]$ for all $f \in \CP{0}$.
For the SDE~\eqref{LW_SDEs}, we have $P^W_tf(x,y)=\E[f(x+\frac{\rho}{\sigma}(Y^y_t-y),Y^y_t)]$, where $Y^y$ is the solution of~\eqref{log-Heston_SDEs}. Thus, it amounts to simulate exactly the $Y^y_t$, and we refer to~\cite[Section 3.1]{AA_book} for a presentation of different exact simulation methods.

However, in practice, the exact simulation of the Cox-Ingersoll-Ross process is longer than a simple Gaussian random variable, and it can be interesting to use an approximation scheme. We use here the one introduced by Ninomiya and Victoir~\cite{NV}. Following Theorem 1.18 in~\cite{AA_MCOM}, we rewrite $\mcL_W=\mcL_0+\mcL_1$ where 
\begin{equation}\label{def_L0L1}
  \mcL_0 =(a-\frac{\sigma^2}{4}-by)(\frac{\rho}{\sigma}\partial_x + \partial_y ),\qquad \mcL_1 =\frac{y}{2}(\rho^2\partial^2_x + 2\rho\sigma\partial_x\partial_y+\sigma^2\partial^2_y) + \frac{\rho \sigma}{4}\partial_x+ \frac{\sigma^2}4 \partial_y,
\end{equation}
are the infinitesimal generator respectively associated to
$$\begin{cases} dX_t &= (\frac{\rho }{\sigma}(a-\sigma^2/4)-\frac{\rho b}{\sigma}Y_t) dt\\
dY_t &= (a-\sigma^2/4-bY_t) dt  \end{cases} \text{ and }
\begin{cases} dX_t &= \frac{\rho \sigma }{4} dt +\rho \sqrt{Y_t} dW_t\\
  dY_t &= \frac{\sigma^2}4 dt +\sigma \sqrt{Y_t} dW_t. \end{cases}$$
Let $\psi_b(t)=\frac{1-e^{-bt}}b$ (convention $\psi_b(t)=t$ for $b=0$) and define
\begin{align}
  \varphi_0(t,x,y)&=\Big(x-\frac{\rho b}{\sigma}\psi_b(t)y +\frac{\rho}{\sigma}\psi_b(t)(a-\frac{\sigma^2}{4}),~e^{-bt}y+\psi_b(t)(a-\frac{\sigma^2}{4})\Big), \label{varphi0} \\
  \varphi_1(t,x,y)&=\Big(x+\frac{\rho }{\sigma}\big((\sqrt{y}+\frac{\sigma t}{2})^2-y\big),~(\sqrt{y}+\frac{\sigma t}{2})^2\Big). \label{varphi1}
\end{align}
We have for $t \ge 0$ and $f \in \CP{0}$, 
\begin{equation}\label{def_P0P1}
  P^0_t f(x,y)= f(\varphi_0(t,x,y)) \text{ and } P^1_t f(x,y)= \E[f(\varphi_1(\sqrt{t}G,x,y))], \text{ with } G\sim \mathcal{N}(0,1).  
\end{equation}
Indeed, $\varphi_0$ is the exact solution of the ODE associated to~$\mathcal{L}_0$, starting from $(x,y)$ at time~$0$, and we can show by Itô calculus that $\varphi_1(W_t,x,y)$ is an exact solution of the SDE associated to~$\mathcal{L}_1$, starting from $(x,y)$ at time~$0$, and with the Brownian motion $d\tilde{W}_t=\text{sgn}\left(\sqrt{y}+\frac \sigma 2 W_t \right)dW_t$. The Ninomiya-Victoir scheme~\cite{NV} for $\mathcal{L}_W$ is then $P^0_{t/2}P^1_tP^0_{t/2}$, and we define \begin{equation}\label{def_PNV}
  \hat{P}_t^{NV}=P^B_{t/2}P^0_{t/2}P^1_tP^0_{t/2} P^B_{t/2}, \text{ when } \sigma^2\le 4a.
\end{equation}
This scheme is well defined only for $\sigma^2\le 4a$, otherwise $\varphi_0$ may send the $y$ component to negative values, and the composition is not well defined. This problem was pointed in~\cite{AA_MCOM} that introduces a second order scheme for any $\sigma\ge 0$. For this scheme,  the normal variable~$G$ in~\eqref{def_P0P1} is replaced by a bounded random variable that matches the five first moments of~$G$ and besides, an ad hoc discrete scheme is used in the neighbourhood of~$0$. However, as indicated in the introduction of this subsection, this prevents us with our analysis to get~\eqref{H2_bar} and thus approximations of higher order. This is why we only consider the Ninomiya-Victoir scheme here.

We now state  the main theorem of the paper.
\begin{theorem}\label{main_theorem}
 Let $\hat{P}_t$ be either $\hat{P}_t^{Ex}$ defined by~\eqref{def_PEx} or $\hat{P}_t^{NV}$ by~\eqref{def_PNV}. Let $T>0$, $n\in \N^*$ and $h_l=T/n^l$. Let $\cPh^{1,n}=\hat{P}_{h_1}^{[n]}$, $\cPh^{2,n}$ be defined by~\eqref{def_boost2} and  $\cPh^{\nu,n}$ the further approximations developed in~\cite{AB}. Let $\nu \ge 1$.  For any $f\in \CP{12\nu}$ $x\in \R$ and $y\ge 0$, we have
  $$\cPh^{\nu,n}f(x,y)-P_T f(x,y) = O(1/n^{2\nu}).$$ 
\end{theorem}
\begin{proof}
  Property~\eqref{H1_bar} is proved in Corollary~\ref{cor_H1} and~\eqref{H2_bar} in Lemma~\ref{lem_H2_estimate}. For $f\in \CP{12\nu}$, there exists $L$ such that $f\in \CPL{12\nu}{L}$. Let $\nu=1$. We get from~\eqref{eq_TT}, 
  $$\|\cPh^{1,n}f-P_Tf\|_{0,L}= \|\sum_{i=0}^{n-1}\hat{P}_{h_1}^{[n-(i+1)]}(\hat{P}_{h_1}-P_{h_1})P_{i{h_1}}f\|_{0,L}\le C^3T \|f\|_{12,L+3} \left(\frac{T}{n} \right)^{2} , $$
  since $\| \hat{P}_{h_1}^{[n-(i+1)]}(\hat{P}_{h_1}-P_{h_1})P_{i{h_1}}f \|_{0,L} \le C\|(\hat{P}_{h_1}-P_{h_1})P_{i{h_1}}f\|_{0,L} \le C^2\|P_{i{h_1}}f\|_{12,L+3} h_1^3 \le C^3\|f\|_{12,L+3} h_1^3$ by using~\eqref{H2_bar}, then~\eqref{H1_bar} and again~\eqref{H2_bar}. This shows the claim for $\nu=1$. For $\nu=2$ (resp. $\nu \ge 3$), the claim is a consequence of~\eqref{boost2_estimate} (resp.~\eqref{def_Pnu}).
\end{proof}
\red{\begin{remark}\label{rk_regularity}
The regularity assumption on~$f$ required by Theorem~\ref{main_theorem} is rather strong, but needed with our analysis. In our numerical experiments of Section~\ref{Sec_num}, we however observe similar rates of convergence for some functions that do not satisfy the assumption of Theorem~\ref{main_theorem}. It would thus be interesting to relax this regularity assumption. In this direction, Rey~\cite{Rey} has recently extended the results of Alfonsi and Bally~\cite{AB} to bounded measurable functions for the Euler scheme under a H\"ormander condition, which ensures sufficient regularization property from the scheme. In the specific case of the log-Heston model, we could use for $|\rho|<1$ the argument of Romano and Touzi~\cite{RoTo} to regularize the payoff. This is left for further research. 
\end{remark}
}

\section{Proof of the main result}\label{Sec_proof}

\subsection{Preliminary results on the norm}

The next lemma gathers basic properties of the family of norms $\|\cdot\|_{k,L}$ defined in Equation~\eqref{def_NormCpolKL}.

\begin{lemma}\label{basic_lemma_CPL}
  Let $k,L\in\N$. We have the following basic properties:
  \begin{enumerate}
    \item $\CPL{k+1}{L}\subset \CPL{k}{L}$. For $f\in \CPL{k+1}{L}$, we have $\|f\|_{k,L}\le \|f\|_{k+1,L}$.
    \item Let $k,\a',\b'\in\N$. For $f\in\CPL{k+\a'+2\b'}{L}$ one has $\|\partial_x^{\a'} \partial_y^{\b'}  f\|_{k,L}\le \|f\|_{k+\alpha'+2\beta',L}$.
    \item $\CPL{k}{L}\subset \CPL{k}{L+1}$ and $\|f\|_{k,L+1}\le 3\|f\|_{k,L}$ for $f\in \CPL{k}{L}$.
    \item Let $\mcM_1$ be the operator defined by $\mcM_1 f(x,y)=yf(x,y)$. For $f\in \CPL{k}{L}$, we have  $\mcM_1f\in \CPL{k}{L+1}$  and $\| \mcM_1f\|_{k,L+1}\le \frac{3}{2}(k+1) \|f\|_{k,L}$.
    \item Let $\mcL$, $\mcL_B$, $\mcL_W$, $\mcL_0$ and $\mcL_1$ the infinitesimal generators defined in Equations~\eqref{log_Heston2_diff-op}, \eqref{def_LB}, \eqref{def_LW} and~\eqref{def_L0L1}. Then, for all $k\in \N$, there exists a constant $C(k)$ such that 
    \begin{align*} &\forall L\in \N, f\in\CPL{k+4}{L}, \ \|\mcL f\|_{k,L+1} + \|\mcL_W f\|_{k,L+1} + \|\mcL_1 f\|_{k,L+1}  \le C(k) \|f\|_{k+4 ,L}, \\
      &\forall L\in \N, f\in\CPL{k+2}{L}, \ \|\mcL_B f\|_{k,L+1}+\|\mcL_0 f\|_{k,L+1} \le C(k)  \|f\|_{k+2 ,L}.\end{align*}
  \end{enumerate}
\end{lemma}
\begin{proof}
  Property (1)-(2) are straightforward. We prove only (3), (4) and (5).

  \noindent(3) Let $a>0$ than $a^{2L}\le 1+a^{2L+2}$. So, we get  $\bff_L(x,y)=1+x^{2L}+y^{2L}\le 3(1+x^{2(L+1)}+y^{2(L+1)})=3\bff_{L+1}(x,y)$ and then $\frac{1}{\bff_{L+1}(x,y)}\le 3 \frac{1}{\bff_{L}(x,y)}$. This gives immediately the claim.

  \noindent(4) Let $f\in \CPL{k}{L}$.  Applying Leibniz rule, one obtains $\partial^\a_x\partial_y^\b \mcM_1f = \b\partial^\a_x\partial_y^{\b-1}f +\mcM_1\partial^\a_x\partial_y^\b f$ for $\a,\b \in \N$. Now, we write 
  \begin{align*}
    \frac{|\partial^\a_x\partial_y^\b[yf(x,y)]|}{\bff_{L+1}(x,y)}&\le \frac{\b|\partial^\a_x\partial_y^{\b-1} f(x,y)|}{\bff_{L+1}(x,y)} + \frac{y|\partial^\a_x\partial_y^\b f(x,y)|}{\bff_{L+1}(x,y)} \\
    &\le  \frac{3\b|\partial^\a_x\partial_y^{\b-1} f(x,y)|}{\bff_{L}(x,y)} +\frac{3|\partial^\a_x\partial_y^\b f(x,y)|}{2 \bff_{L}(x,y)},
  \end{align*}
  where we used the comparison above between $\bff_L$ and $\bff_{L+1}$  for the first term and, for the second term,  
  $y \bff_L(x,y) =y+y^{2L+1}+yx^{2L}+\le 1+y^{2L+2}+\frac{1+y^2}{2}x^{2L}\le \frac 32  \bff_{L+1}(x,y)$ by using $y+y^{2L+1}\le 1+y^{2L+2}$ and then Young's inequality. Then, we obtain 
  \begin{align*}
    \|\mcM_1 f\|_{k,L+1} &\le \sum_{\a+2\b\le k} \left(3\beta \sup_{(x,y)\in \R \times \R_+}\frac{3\b|\partial^\a_x\partial_y^{\b-1} f(x,y)|}{\bff_{L}(x,y)}  + \frac 3 2 \sup_{(x,y)\in \R \times \R_+}\frac{3\b|\partial^\a_x\partial_y^{\b} f(x,y)|}{\bff_{L}(x,y)} \right) \\
    &\le 3\left( \lfloor k/2 \rfloor +1/2 \right) \| f\|_{k,L}.
  \end{align*}
  \noindent(5) We prove only the estimate for $\mcL$, the others are obtained using the same arguments. We have $\|\mcL f\|_{k,L+1}\le \frac{1}{2}\|\mcM_1(\partial_x^2+2\rho\sigma\partial_x\partial_y+\sigma^2\partial_y^2-\partial_x-2b\partial_y) f\|_{k,L+1} + \|(r\partial_x+a\partial_y) f\|_{k,L+1}$, by using linearity of derivatives and the triangular inequality.  
  We conclude using property (4) and (2) for the first term,  (2) and (3) for the second and finally property number (1) to upper bound $\|\mcL f\|_{k,L+1}$ by $C(k)\|f\|_{k+4,L+1}$, where $C(k)$ is  a constant  depending on $k$ and on the parameters ($r,\rho,a,b$ and $\sigma$).
\end{proof}

\subsection{The Cauchy problem of the Log-Heston SDE}

In this subsection, we aim at proving the estimates on the derivatives of the Cauchy problem.  The representation formula presented below is a result of Briani, Caramellino and Terenzi~\cite{BrCaTe}.

\begin{prop}\label{prop-rep-logHeston-estim}
  Let $k,L \in \mathbb{N}$ and suppose that $f \in \CPL{k}{L}$. Let $\lambda \ge 0$, $c,d \in \R$. Let $(X^{t, x, y}, Y^{t, y})$ be the solution to the SDE, for $s\ge t$, 
  \begin{equation}\label{log-Heston-ext_SDEs}
    \begin{cases}
      X^{t,x,y}_s &= x +\int_t^s (c+dY^y_r) dr + \int_t^s  \lambda \sqrt{Y^y_r} (\rho dW_r + \sqrt{1-\rho^2} dB_r)\\
      Y^{t,y}_s &= y+ \int_t^s (a-bY^y_r) dr +\sigma \int_t^s  \sqrt{Y^y_r} dW_r,
    \end{cases}
  \end{equation}
 and set
  $$ u(t, x, y)=\E[f(X_{T}^{t, x, y}, Y_{T}^{t, y})]=P_{T-t}f(x,y).$$
  Then, $u(t,\cdot,\cdot)\in \CPL{k}{L}$ and the following stochastic representation holds for $\a+2\beta \le k$,
  \small
  \begin{multline}\label{stoc_repr-new}
  \partial_{x}^{\a} \partial_{y}^{\b} u(t, x, y)=\E\bigg[e^{-\b b(T-t)} \partial_{x}^{\a} \partial_{y}^{\b} f\big(X_{T}^{\b, t, x, y}, Y_{T}^{\b,t, y}\big) \\
   \quad+\b \int_{t}^{T}e^{-\b b(s-t)}\Big(\frac{\lambda^2}{2} \partial_{x}^{\a+2} \partial_{y}^{\b-1} u + d \partial_{x}^{\a+1} \partial_{y}^{\b-1} u \Big)\big(s, X_{s}^{\b, t, x, y}, Y_{s}^{\b,t, y}\big) d s\bigg],
  \end{multline}
  where $\partial_{x}^{\a} \partial_{y}^{\b-1}  u:=0$ when $\beta=0$ and $(X^{\b, t, x, y}, Y^{\b, t, y}), \beta \geq 0$, denotes the solution starting from $(x, y)$ at time $t$ to the SDE \eqref{log-Heston-ext_SDEs}  with parameters
  \begin{equation}\label{parameters-new}
    \rho_{\b}=\rho,  \quad a_{\b}=a+\b \frac{\sigma^{2}}{2}, \quad b_{\beta}= b, \quad c_{\b}=r+\b \rho \sigma \lambda, \quad d_{\b}=d, \quad \sigma_{\b}=\sigma.
  \end{equation}
  Moreover, one has the following norm estimation for the semigroup   \begin{equation}
   \forall k, L \in \N, T>0, \  \exists C, \forall f \in  \CPL{k}{L}, t\in[0,T], \   \|P_tf\|_{k,L}\le  \|f\|_{k,L} e^{Ct}.
  \end{equation}
\end{prop}

\begin{proof}
Proposition~\ref{prop-rep-logHeston-estim} basically restates~\cite[Proposition 5.3]{BrCaTe} in our framework (note that our space $\CPL{k}{L}$ already includes twice more derivatives in~$x$ than in $y$ and that we have added the scaling factor $\lambda$ for convenience). The only additional result is the norm estimate, which we prove now. 

Let $f \in  \CPL{k}{L}$. We will prove that for all $(\a,\b)$ such that $\alpha+2 \b \le k$, there exists a constant $C\in \R_+$ such that 
\begin{equation}\label{inter_estim}
  \sup_{x\in \R y \in \R_+} \frac{|\partial_x^\alpha \partial_y^\beta u(t,x,y)|}{\bff_L(x,y)} \le \sup_{x\in \R y \in \R_+} \frac{|\partial_x^\alpha \partial_y^\beta f(x,y)|}{\bff_L(x,y)}e^{C(T-t)}  + C \|f\|_{k,L}(T-t).
\end{equation}
Let us note that this implies $\sup_{x\in \R y \in \R_+} \frac{|\partial_x^\alpha \partial_y^\beta u(t,x,y)|}{\bff_L(x,y)} \le \tilde{C} \|f\|_{k,L}$, with $\tilde{C}=e^{CT}+CT$.

For $\beta=0$, the estimate is straightforward : from~\eqref{stoc_repr-new} and $f\in\CPL{k}{L}$, one has
$$
|\partial_{x}^{\a} u(t,x,y)| \le  \left(\sup_{x'\in \R y' \in \R_+} \frac{|\partial_x^\alpha f(x',y')|}{\bff_L(x',y')}  \right)\E \big[\bff_L(X^{t,x,y}_T,Y_T^{t,y})\big],
$$
and we get easily~\eqref{inter_estim} by using the estimate on moments~\eqref{eq_moments} that we prove below. 

Suppose now that the estimate \eqref{inter_estim} is true up to $\beta-1$, and let us prove it for $\beta$.
Using \eqref{stoc_repr-new}  and the induction hypothesis for the integral term, we get
\begin{align*}
  |\partial_{x}^{\a} \partial_{y}^{\b} u(t, x, y)|\le & 
  e^{-\beta b (T-t)} \left(\sup_{x'\in \R y' \in \R_+} \frac{|\partial_x^\alpha \partial_y^\beta f(x',y')|}{\bff_L(x',y')} \right)
  \E\left[ \bff_L(X^{\beta,t,x,y}_T,Y^{\beta,t,y}_T)  \right]\\
  &+\beta\frac{\lambda^2+|d|}{2}e^{\beta |b| T}  \tilde{C} \|f\|_{k,L} \int_t^T \E[\bff_L(X^{\beta,t,x,y}_s,Y^{\beta,t,y}_s)] ds.
\end{align*}
This gives~\eqref{inter_estim} by using again the estimate on the moments~\eqref{eq_moments}. This shows~\eqref{inter_estim} by induction, and we finally sum~\eqref{inter_estim} over $\alpha+2\beta\le k$ to get
$$\|P_{T-t}f\|_{k,L}\le \|f\|_{k,L} e^{C(T-t)}+C(k+1)^2 \|f\|_{k,L} (T-t)\le \|f\|_{k,L} e^{C(1+(k+1)^2)(T-t)},$$
proving the claim.
\end{proof}

\begin{lemma}\label{lem_moments}
  Let $(X^{x,y},Y^y)$ be the solution of~\eqref{log-Heston-ext_SDEs} starting from $(x,y)$ at time~$0$. Let $T>0$. For any $L\in \N$, there is a constant $C\in \R_+$ depending on $T$, $L$ and the SDE parameters, such that 
\begin{equation}\label{eq_moments}\E[\bff_L(X^{x,y}_t,Y^y_t)]\le e^{Ct}\bff_L (x,y), \ t\in [0,T].
\end{equation}
\end{lemma}

\begin{proof}
  We use the affine (and thus polynomial) property of the extended log-Heston process~\eqref{log-Heston-ext_SDEs}, see~\cite[Example 3.1]{CKRT}. By~\cite[Theorem 2.7]{CKRT}, we get that the log-Heston semigroup acts as a matrix exponential on the polynomial functions of degree lower than $2L$. This gives $\E[(X^{x,y}_t)^{2L}+(Y^y_t)^{2L}]=x^{2L}+y^{2L}+\sum_{i+j\le 2L}\varphi_{i,j}(t) x^i y^j$, with $\varphi_{i,j}\in \mathcal{C}^1([0,T],\R)$ such that $\varphi_{i,j}(0)=0$. Using that $|x|^i y^j\le |x|^{i+j}+ y^{i+j}\le 2 \bff_L(x,y)$ for $i+j\le 2L$ and using that $\varphi'_{i,j}$ is bounded on $[0,T]$, we get 
  $$\E[\bff_L(X^{x,y}_t,Y^y_t)]\le \bff_L(x,y)+ Ct \bff_L(x,y)\le \bff_L(x,y)e^{Ct}, $$
  with $C=2 \sum_{i,j\le 2L } \max_{[0,T]}|\varphi'_{i,j}|$. 
\end{proof}

\subsection{Proof of~\eqref{H1_bar} and ~\eqref{H2_bar}}

We start by proving the property~\eqref{H2_bar} in the next lemma.
\begin{lemma}\label{lem_H2_estimate}
  Let $t\in[0,T]$,  $k,L\in\N$ and $f\in\CPL{k}{L}$. Let $\varphi_0 $ be the function defined in Equation~\eqref{varphi0}. Then,  there exists $C$ such that, for $I\in \{0,1,B,W\}$, $\|P^{I}_tf\|_{k,L}\le e^{Ct}\|f\|_{k,L}$, for $t\in [0,T]$. The semigroup approximations $\hat{P}^{Ex}_t$ and $\hat{P}^{NV}_t$ enjoy the same property and satisfy~\eqref{H2_bar}.
\end{lemma}
\begin{proof}
We apply four times Proposition~\ref{prop-rep-logHeston-estim} with 
\begin{itemize}
  \item $\tilde{a}=a-\frac{\sigma^2}4$, $\tilde{b}
=b$, $\tilde{c}=\frac{\rho}{\sigma}\left(a-\frac{\sigma^2}4\right) $, $\tilde{d}=-b\frac{\rho}{\sigma}$, $\tilde{\lambda}=0$, $\tilde{\sigma}=0$ for $P^0$,
\item  $\tilde{a}=\frac{\sigma^2}4$, $\tilde{b}
=0$, $\tilde{c}=\frac{\rho \sigma}4 $, $\tilde{d}=0$, $\tilde{\lambda}=\rho$, $\tilde{\sigma}=\sigma$, $\tilde{\rho}=1$ for $P^1$,
\item $\tilde{a}=0$, $\tilde{b}
=0$, $\tilde{c}=r-\frac{\rho a}{\sigma} $, $\tilde{d}=\frac{\rho b}{\sigma} -\frac 12$, $\tilde{\lambda}=\bar{\rho}$,  $\tilde{\sigma}=0$, $\tilde{\rho}=0$ for $P^B$,
\item $\tilde{a}=a$, $\tilde{b}
=b$, $\tilde{c}=\frac{\rho a}\sigma $, $\tilde{d}=-\frac{\rho b}\sigma$, $\tilde{\lambda}=\rho$, $\tilde{\sigma}=\sigma$, $\tilde{\rho}=1$ for $P^W$,
\end{itemize}
where the tilde parameters are the ones used in Equation~\eqref{log-Heston-ext_SDEs}. This gives the first claim. Then, we deduce easily that $\|\hat{P}^{Ex}_tf\|_{k,L}\le e^{Ct/2} \|P^W_{t} P^B_{t/2}f\|_{k,L}\le e^{2Ct} \|f\|_{k,L}$ by using twice the estimate for $P^B$ and once for $P^W$. Similarly, we obtain $\|\hat{P}_t^{NV}f\|_{k,L}\le e^{3Ct} \|f\|_{k,L}$, by using the estimates for $P^B$, $P^0$ and $P^1$.

Now, the property~\eqref{H2_bar} follows easily: consider $l\ge 1$ and $Q_l=\hat{P}^{NV}_{\frac{T}{n^l}}$, we have for $f \in\CPL{k}{L}$, $\|Q_l f\|_{k,L}\le e^{3C \frac{T}{n^l}}$ and thus for any $j\le n^l$, $\|Q^{[j]}_l f\|_{k,L}\le e^{3C \frac{jT}{n^l}}\le e^{3CT}$, which gives~\eqref{H2_bar}. The same result holds for $\hat{P}^{Ex}$.  
\end{proof}

We now turn to the proof of the property~\eqref{H1_bar}. We first state a general result on the composition of approximation schemes that fits our framework with the norm family $\|\cdot\|_{k,L}$.  It can be seen as a variant of~\cite[Proposition 2.3.12]{AA_book} and says, heuristically, that the composition of schemes works as a composition of operators. 

\begin{lemma}\label{lem_compo}(Scheme composition). Let $\nu \in \N$ and $T>0$. Let $V_i$, $i\in \{1,\dots,I\}$,  be infinitesimal generators such that  there exists $k_i,L_i \in \N$ such that
\begin{equation}\label{estim_V}
  \forall k \in \N, \exists C\in \R_+,  \forall f \in \CPL{k+k_i}{L}, V_i f \in  \CPL{k}{L+L_i} \text{ and } \|V_if\|_{k,L+L_i}\le C \|f\|_{k+k_i,L}.
\end{equation}
Let $k^\star=\max_{1\le i\le I} k_i$ and $L^\star=\max_{1\le i\le I} L_i$.
We assume that for any $i$, $\hat{P}^i_t:\CPL{0}{L} \to \CPL{0}{L} $ is such that
\begin{align} \forall k, L \in \N, 0\le  \bar{q}\le \nu+1, &\ \exists C,\ \forall f \in \CPL{k+\bar{q}k_i}{L}, \forall t\in[0,T], \notag \\& \|\hat{P}^i_t f -\sum_{q=0}^{\bar{q}-1}\frac{t^q}{q!} V_i^qf\|_{k, L+ \bar{q} L_i} \le C t^{\bar{q}} \|f\|_{k+\bar{q} k_i, L}. \label{assump_Phat}\end{align}
Then, we have for $\lambda_1,\dots,\lambda_I\in [0,1]$, 
\begin{align}
 & \forall k, L \in \N, 0\le  \bar{q}\le \nu+1, \ \exists C, \forall f \in \CPL{k+\bar{q}k^\star}{L}, \forall t \in [0,T]  \notag \\
  & \left\|\hat{P}^I_{\lambda_I t} \dots \hat{P}^1_{\lambda_1 t} f -\sum_{q_1+\dots+q_I\le \bar{q}-1}\frac{\lambda_1^{q_1}\dots \lambda_I^{q_I} t^{q_1+\dots+q_I} }{q_1!\dots q_I!} V_I^{q_I}\dots V_1^{q_1}f\right\|_{k, L+ \bar{q} L^\star} \le C t^{\bar{q}} \|f\|_{k+\bar{q} k^\star, L}. \label{compo_oper}
\end{align}
\end{lemma}
\begin{proof} For readability, we make the proof with $I=2$ operators. 
  Let $\bar{q}\le \nu+1$ and $f\in \CPL{k+\bar{q}k^\star}{L}$. We define $R^1f=\hat{P}^1_{\lambda_1 t} f -\sum_{q_1=0}^{\bar{q}-1}\frac{\lambda_1^{q_1} t^{q_1}}{q_1!} V_1^{q_1}f$. For $t\in [0,T]$, we have $\lambda_1 t \in [0,T]$  since $\lambda_1 \in [0,1]$ and by assumption~\eqref{assump_Phat}, we have $R^1 f \in \CPL{k+\bar{q} k^\star - \bar{q} k_1 }{L+\bar{q}L_1}$ and 
  $$\|R^1 f\|_{k+\bar{q}k^\star - \bar{q} k_1,L+\bar{q}L_1}\le C t^{\bar{q}} \|f\|_{k+\bar{q}k^\star,L}.$$  
  We now write 
  $$ \hat{P}^2_{\lambda_2 t} \hat{P}^1_{\lambda_1 t} f= \sum_{q_1=0}^{\bar{q}-1}\frac{\lambda_1^{q_1} t^{q_1}}{q_1!} \hat{P}^2_{\lambda_2 t} V_1^{q_1}f + \hat{P}^2_{\lambda_2 t} R^1 f.$$
 Since $V_1^{q_1}f \in \CPL{k+\bar{q}k^\star-q_1k_1}{L+q_1L_1}$, we  apply~\eqref{assump_Phat} to get $$\hat{P}^2_{\lambda_2 t} V_1^{q_1}f = \sum_{q_2=0}^{\bar{q}-q_1-1} \frac{\lambda_2^{q_2} t^{q_2}}{q_2!} V_2^{q_2}V_1^{q_1}f + R^2_{q_1}f,$$ with
 $\|R^2_{q_1} f \|_{k+\bar{q}k^\star-q_1k_1-(\bar{q}-q_1)k_2,L+q_1L_1+(\bar{q}-q_1)L_2}\le C t^{\bar{q}-q_1}\|f\|_{k+\bar{q}k^\star,L}$ by~\eqref{estim_V} and~\eqref{assump_Phat}. We also have $\|\hat{P}^2_{\lambda_2 t} R^1 f\|_{k+\bar{q}k^\star - \bar{q} k_1,L+\bar{q}L_1}\le C t^{\bar{q}} \|f\|_{k+\bar{q}k^\star,L}$ by~\eqref{assump_Phat}. Since  
 $$k+\bar{q}k^\star-q_1k_1-(\bar{q}-q_1)k_2\ge k, \ L+q_1L_1+(\bar{q}-q_1)L_2\le L+\bar{q} L^\star, $$
 for all $0\le q_1\le \bar{q}-1$, and using Lemma~\ref{basic_lemma_CPL} (1 and 3), we get
 $$ \left\|\hat{P}^2_{\lambda_2 t} \hat{P}^1_{\lambda_1 t} f -\sum_{q_1+q_2\le \bar{q}-1}\lambda_1^{q_1}\lambda_2^{q_2}\frac{t^{q_1}t^{q_2}}{q_1!q_2!} V_2^{q_2} V_1^{q_1}f \right\|_{k, L+ \bar{q} L^\star} \le C t^{\bar{q}} \|f\|_{k+\bar{q} k^\star, L}. \qedhere$$
\end{proof}

\begin{lemma}\label{lem_expan_CPL_scheme} 
   Let $L_0=L_1=L_B=L_W=L_H=1$, $k_0=k_B=2$, $k_1=k_W=k_H=4$. Let denote $\mcL_H=\mcL$ and $P^H_t=P_t$ the log-Heston semigroup.  Let $i\in \{0,1,B,W,H\}$. We have 
  $$ \forall k,L \in \N, \exists C\in \R_+, \forall f \in \CPL{k+k_i}{L}, \  \|\mathcal{L}_if\|_{k,L+L_i}\le C \|f\|_{k+k_i,L}. $$ 
  Besides, for any $\bar{q}\in \N$, we have 
  \begin{align*} \forall k, L \in \N,  \ \exists C,\ \forall f \in \CPL{k+\bar{q}k_i}{L}, \notag  \|P^i_t f -\sum_{q=0}^{\bar{q}-1}\frac{t^q}{q!} \mathcal{L}_i^qf\|_{k, L+ \bar{q} L_i} \le C t^{\bar{q}} \|f\|_{k+\bar{q} k_i, L}. \end{align*}
\end{lemma}
\begin{proof}
  The first part of the statement is proved in Lemma~\ref{basic_lemma_CPL}. For $\bar{q}=0$, the estimate is simply the one given by Lemma~\ref{lem_H2_estimate} (or Proposition~\ref{prop-rep-logHeston-estim} for $P^H_t$).

  We now consider $\bar{q}\ge 1$. As already pointed in the proof of Lemma~\ref{lem_H2_estimate}, each operator is the infinitesimal generator of~\eqref{log-Heston-ext_SDEs} with a suitable choice of parameter. Then, by applying Itô's formula and taking the expectation, we get $P^{i}_tf(x,y)=f(x,y)+\int_0^t P^i_s \mathcal{L}_if(x,y)ds$. By iterating,  we obtain for $f \in \CPL{k+\bar{q}k_i}{L}$,
  \begin{equation}\label{expan_P_general}
    P^{i}_tf(x,y) = \sum_{j=0}^{\bar{q}-1} \frac{t^j}{j!} \mcL_{i}^jf(x,y) + \int_0^{t} \frac{(t-s)^{\bar{q}-1}}{(\bar{q}-1)!} P^{i}_{s} \mcL^{\bar{q}}_{i}f(x,y)ds.  
  \end{equation}
We have $\|\mcL^{\bar{q}}_{i}f\|_{k, L+\bar{q} L_i}\le C^{\bar{q}} \|f\|_{k+\bar{q} k_i, L}$ by Lemma~\ref{basic_lemma_CPL}-(5) and thus $\| P^i_s \mcL^{\bar{q}}_{i}f \|_{k, L+\bar{q} L_i}\le C^{q+1} \| f \|_{k+\bar{q} k_i, L}$ for $s\in [0,T]$, by using the result for $\bar{q}=0$.  Therefore, we get by the triangle inequality
\begin{align*}
  \left\|P^i_tf-\sum_{j=0}^{\bar{q}-1} \frac{t^j}{j!} \mcL_{i}^jf\right\|_{k,L+\bar{q}L_i}\le  \int_0^t \frac{(t-s)^{\bar{q}-1}}{(\bar{q}-1)!}C^{\bar{q}+1} \| f \|_{k+\bar{q} k_i, L} ds= \frac{t^{\bar{q}}}{\bar{q}!}C^{\bar{q}+1} \| f \|_{k+\bar{q} k_i, L}.\quad \qedhere
\end{align*}
\end{proof}

\begin{corollary}\label{cor_H1}
  Let $T>0$. Let $\hat{P}_t$ denote either $\hat{P}_t^{Ex}$ or $\hat{P}_t^{NV}$. We have, for $\bar{q}\le 3$,  
  \begin{align*} \forall k, L \in \N,  \ \exists C,\ \forall f \in \CPL{k+4\bar{q}}{L}, \forall t \in [0,T],  \|\hat{P}_t f -\sum_{q=0}^{\bar{q}-1}\frac{t^q}{q!} \mathcal{L}^qf\|_{k, L+ \bar{q} } \le C t^{\bar{q}} \|f\|_{k+4\bar{q}, L}, \end{align*}
  and~\eqref{H1_bar} holds.
\end{corollary}
\begin{proof}
We prove the result for $\hat{P}_t^{Ex}$, the argument is similar for $\hat{P}^{NV}_t$.   
We use Lemma~\ref{lem_expan_CPL_scheme} for $P^W_t$ and $P^B_t$ and apply then Lemma~\ref{lem_compo}. For $\bar{q}=0,1,2$, we get easily the claim. For $\bar{q}=3$, we get since $\hat{P}^{Ex}_t=P^B_{t/2}P^W_{t/2}P^B_{t/2}$, 
$$\left\|\hat{P}^{Ex}_t f -\sum_{q_1+q_2+q_3\le 2} \frac{(1/2)^{q_1+q_3} t^{q_1+q_2+q_3} }{q_1!q_2 q_3!} \mathcal{L}_B^{q_3}\mathcal{L}_W^{q_2}\mathcal{L}_B^{q_1}f\right\|_{k, L+ 3} \le C t^{3} \|f\|_{k+12, L}.$$
The term of order two is 
$$ \frac{1}8 \mathcal{L}_B^{2}f+ \frac{1}4 \mathcal{L}_B^{2}f+\frac{1}8 \mathcal{L}_B^{2}f+  \frac 12 \mathcal{L}_B\mathcal{L}_W+ \frac 12 \mathcal{L}_W\mathcal{L}_B +\frac{1}2 \mathcal{L}_W^{2}f=\frac 12 (\mathcal{L}_B+\mathcal{L}_W)^2f, $$
and thus $\sum_{q_1+q_2+q_3\le 2} \frac{(1/2)^{q_1+q_3} t^{q_1+q_2+q_3} }{q_1!q_2 q_3!} \mathcal{L}_B^{q_3}\mathcal{L}_W^{q_2}\mathcal{L}_B^{q_1}f=\sum_{q=0}^2 \frac {t^q}{q!}(\mathcal{L}_B+\mathcal{L}_W)^qf=\sum_{q=0}^2 \frac {t^q}{q!}\mathcal{L}^qf$.

Now, we use Lemma~\ref{lem_expan_CPL_scheme} to get $\left\|P_t f - \sum_{q=0}^2 \frac {t^q}{q!}\mathcal{L}^qf \right\|_{k, L+ 3} \le C t^{3} \|f\|_{k+12, L}$. The triangular inequality then gives 
$$ \left\|P_t f - \hat{P}^{Ex}_t \right\|_{k, L+ 3} \le C t^{3} \|f\|_{k+12, L},$$
which is precisely~\eqref{H1_bar}.
\end{proof}

\section{Numerical results}\label{Sec_num}

\subsection{Implementation}\label{Subsec_implementation}
We explain in this subsection the implementation of the schemes associated to~$\hat{P}_t^{Ex}$ and $\hat{P}_t^{NV}$, and then of the Monte-Carlo estimator of $\cPh^{\nu,n}$, $\nu \in \{1,2\}$. We will note either $\cPh^{Ex,\nu,n}$ or $\cPh^{NV,\nu,n}$ to emphasize what semigroup approximation is used.

On a single time step, the scheme associated to $\hat{P}_t^{NV}$ is given by
\begin{align*} 
  \hX^{x,y}_t &= x +(r-\frac{\rho}{\sigma}a)t +\frac{\rho}{\sigma}(\hY^y_t-y)  +(\frac{\rho}{\sigma}b-\frac{1}{2})\frac{y+\hY^y_t}{2}t +\sqrt{(1-\rho^2)\frac{t}{2}}\bigg(\sqrt{y}N_1+\sqrt{\hY^y_t}N_2\bigg),\\
  \hY^y_t &= (a-\frac{\sigma^2}{4})\psi_b(\frac{t}{2})+e^{-b\frac{t}{2}}\bigg( \sqrt{(a-\frac{\sigma^2}{4})\psi_b(\frac{t}{2})+e^{-b\frac{t}{2}}y} +\frac{\sigma\sqrt{t}}{2}G \bigg)^2,
\end{align*}
where $N_1,N_2,G$ are three independent random variables with the standard normal distribution. It is obtained from the composition~\eqref{def_PNV} and by using accordingly the maps $\varphi_0$, $\varphi_1$ and $\varphi_B$ that represent the semigroups, see Equations~\eqref{varphiB} and~\eqref{def_P0P1}.

One should remark however that the conditional law of $\hX^{(x,y)}_t$ given  $\hY^y_t$ is normal with mean $x+\frac{\rho}{\sigma}(\hY^y_t-y)+(r-\frac{\rho}{\sigma}b)t +(\frac{\rho}{\sigma}b-\frac{1}{2})\frac{y+\hY^y_t}{2}t$ and variance $t(1-\rho^2)(y+\hY^y_t)/2 $. Therefore, we rather consider the following probabilistic representation, that has the same law and requires to simulate one standard Gaussian random variable $N$ instead of the couple $(N_1,N_2)$ for the first component:
\begin{align*}
  \hX^{x,y}_t &= x +(r-\frac{\rho}{\sigma}a)t +\frac{\rho}{\sigma}(\hY^y_t-y) +(\frac{\rho}{\sigma}b-\frac{1}{2})\frac{y+\hY^y_t}{2}t +\sqrt{(1-\rho^2)\frac{y+\hY^y_t}{2}t}N, \\
  \hY^y_t &= (a-\frac{\sigma^2}{4})\psi_b(\frac{t}{2})+e^{-b\frac{t}{2}}\bigg( \sqrt{(a-\frac{\sigma^2}{4})\psi_b(\frac{t}{2})+e^{-b\frac{t}{2}}y} +\frac{\sigma\sqrt{t}}{2}G \bigg)^2.
\end{align*}
We note $\varphi^{NV}(t,x,y,N,G):=(\hX^{x,y}_t,\hY^{y}_t)$ this map. 
The same trick can be used for~$\hat{P}^{Ex}_t$ when the exact simulation is used for the CIR component, and we define 
$$\varphi_X^{Ex}(t,x,y,N,Y^y_t)=x +(r-\frac{\rho}{\sigma}a)t +\frac{\rho}{\sigma}(Y^y_t-y) +(\frac{\rho}{\sigma}b-\frac{1}{2})\frac{y+Y^y_t}{2}t +\sqrt{(1-\rho^2)\frac{y+Y^y_t}{2}t}N,$$
the map that gives the log-stock component. 

We now explain how to get the Monte-Carlo estimator for $\cPh^{1,n}$ and then $\cPh^{2,n}$. We start with the simulation scheme for $\hat{P}^{Ex}_t$. Let us consider $T>0$, $h_1=T/n$ and the regular time grid $\Pi^0=\{kh_1, \ 0\le k\le n\}$. We simulate exactly $Y_{kh_1}$, $1\le k\le n$, the CIR component starting from $Y_0=y$, and we set
$$\hat{X}^{Ex,0}_{k h_1}=\varphi_X^{Ex,0}(h_1,\hat{X}^{Ex,0}_{(k-1) h_1},Y_{(k-1)h_1},N_k,Y_{kh_1}), \ 1\le k\le n,$$
where $(N_k)_{1\le k\le N}$ are standard normal random variable such that $N_k$ is independent from $(N_{k'})_{k'<k}$ and the process $Y$. 
The Monte-Carlo estimator of $\cPh^{Ex,1,n}$ is then 
$$\frac{1}{M_1}\sum_{m=1}^{M_1} f(\hat{X}^{Ex,0,(m)}_{T},Y^{(m)}_T),$$
where $M_1$ is the number of independent samples. We now present how to calculate the correcting term in $\cPh^{2,n}$. To do so, we draw an independent random variable~$\kappa$ that is uniformly distributed on $\{0,\dots,n-1\}$ and selects the time-step to refine. We note $\Pi^1=\Pi^0 \cup  \{ \kappa h_1 + k' h_2 , 1 \leq k' \leq n-1 \}$ the refined (random) grid, where $h_2=T/n^2$. We simulate exactly~$Y$ on this time grid and define the scheme $\hat{X}^{Ex,1}$ as follows:
\begin{align*}
  &\hat{X}^{Ex,1}_{k h_1}=\hat{X}^{Ex,0}_{k h_1} \text{ for } k\le \kappa,\\
  &\hat{X}^{Ex,1}_{\kappa h_1 +k'h_2}=\varphi_X^{Ex}(h_2,\hat{X}^{Ex,1}_{\kappa h_1 +(k'-1)h_2},Y_{\kappa h_1 +(k'-1)h_2},\tilde{N}_{k'},Y_{\kappa h_1 +k'h_2}), \ 1\le k'\le n, \\
  &\hat{X}^{Ex,1}_{k h_1}=\varphi_X^{Ex,1}(h_1,\hat{X}^{Ex,1}_{(k-1) h_1},Y_{(k-1)h_1},N_k,Y_{kh_1}), \ \kappa+\red{2}< k\le n,
\end{align*}
where $(\tilde{N}_{k'})_{1\le k'\le N}$ are i.i.d. random normal variable, independent from $\kappa$ and $(N_k,Y_{kh_1})_{k\le \kappa}$. We then define the Monte-Carlo estimator of $\cPh^{Ex,2,n}$ (see Eq.~\eqref{boost2_RG}) by
\begin{equation}\label{estimateur_ex}\frac{1}{M_1}\sum_{m=1}^{M_1} f(\hat{X}^{Ex,0,(m)}_{T},Y^{(m)}_T)+\frac{1}{M_2}\sum_{m=1}^{M_2} n\left(f(\hat{X}^{Ex,1,(m)}_{T},Y^{(m)}_T)-f(\hat{X}^{Ex,0,(m)}_{T},Y^{(m)}_T)\right).\end{equation}
Note that we reuse the same Monte-Carlo samples in the two sums as it has been observed in~\cite[Subsection 6.3]{AL} that it is more efficient. The tuning of the parameters $M_1$ and $M_2$ is made to minimize the computational cost to achieve \red{a given statistical precision~$\varepsilon>0$, which gives from~\cite[Eq. (6.11)]{AL}
\begin{align}
  M_{1}  &=\left\lceil \frac{1}{\varepsilon^2}\left( \sigma^2_2(n) + 2\Gamma(n) +\sqrt{\frac 32 \big(\sigma^2_2(n)+2\Gamma(n)\big)\mathbf{V}(n) }  \right)  \right\rceil, \label{def_M1} \\ M_{2} & =\left\lceil \frac{1}{\varepsilon^2}\left( \mathbf{V}(n) +\sqrt{\frac{2}{3}\big(\sigma^2_2(n)+2\Gamma(n)\big)\mathbf{V}(n) }  \right)  \right\rceil,\label{def_M2}
\end{align} 
with $\sigma_2^2(n)=\text{Var}(f(\hat{X}^{Ex,0}_{T},Y_T))$, $\mathbf{V}(n)=\text{Var}\left(n\left(f(\hat{X}^{Ex,1}_{T},Y_T)-f(\hat{X}^{Ex,0}_{T},Y_T)\right)\right)$ and $\Gamma(n)=\text{Cov}\left(f(\hat{X}^{Ex,0}_{T},Y_T),n\left(f(\hat{X}^{Ex,1}_{T},Y_T)-f(\hat{X}^{Ex,0}_{T},Y_T)\right)\right)$. In practice, these quantities are estimated on the first 1000 or 10000 samples. Usually, we get $M_2\le M_1$, and we take $M_1=M_2$ instead of~\eqref{def_M1} otherwise.} Let us stress here that it is important for the variance of the estimator to use the same noise for the simulations of $\hat{X}^{Ex,1}$ and $\hat{X}^{Ex,0}$. In particular, the normal random variable $N_{\kappa+1}$ should depend on $(\tilde{N}_k)_{1\le k\le N}$. A natural choice is to take
$$N_{\kappa+1}=N^{\textup{st}}_{\kappa+1} \text{ where }N^{\textup{st}}=\frac 1{\sqrt{n}}\sum_{k=1}^n\tilde{N}_k,$$
if we think of Brownian increments. We will discuss this choice later on in Subsection~\ref{Subsec_coupling}.

Let us now present the scheme for~$\hat{P}^{NV}_t$, that is well defined for $\sigma^2\le 4a$. The scheme on the coarse grid~$\Pi^0$ is defined by
$$(\hat{X}^{NV,0}_{kh_1},\hat{Y}^{NV,0}_{kh_1})=\varphi^{NV}(h_1,\hat{X}^{NV,0}_{(k-1)h_1},\hat{Y}^{NV,0}_{(k-1)h_1},N_k,G_k), \ 1\le k\le n,$$
where $N_k,G_k$, $1\le k\le n$, are two independent standard normal variables independent of $(N_{k'},G_{k'})_{k'<k}$. The Monte-Carlo estimator of~$\cPh^{NV,1,n}$ is then $\frac{1}{M_1}\sum_{m=1}^{M_1} f(\hat{X}^{NV,0,(m)}_{T},\hat{Y}^{NV,0,(m)}_T)$. The scheme on the refined random grid~$\Pi^1$ is defined by 
\begin{align*}
  &(\hat{X}^{NV,1}_{k h_1},\hat{Y}^{NV,1}_{k h_1})=(\hat{X}^{NV,0}_{k h_1},\hat{Y}^{NV,0}_{k h_1}) \text{ for } k\le \kappa,\\
  &(\hat{X}^{NV,1}_{\kappa h_1 +k'h_2},\hat{Y}^{NV,1}_{\kappa h_1 +k'h_2})=  \varphi^{NV}(h_2,\hat{X}^{NV,1}_{\kappa h_1 +(k'-1)h_2},\hat{Y}^{NV,1}_{\kappa h_1 +(k'-1)h_2},\tilde{N}_{k'},\tilde{G}_{k'}), \ 1\le k'\le n, \\
  &(\hat{X}^{NV,1}_{k h_1},\hat{Y}^{NV,1}_{k h_1})=\varphi^{NV}(h_1,\hat{X}^{NV,1}_{(k-1)h_1},\hat{Y}^{NV,1}_{(k-1)h_1},N_k,G_k), \ \kappa+\red{2}< k\le n,
\end{align*}
where $(\tilde{N}_{k'},\tilde{G}_{k'})_{1\le k'\le N}$, are independent standard normal variables that are also independent of $\kappa$ and $(N_k,G_{k})_{k\le \kappa}$. The Monte-Carlo estimator of $\cPh^{NV,2,n}$  is then defined by
$$\frac{1}{M_1}\sum_{m=1}^{M_1} f(\hat{X}^{NV,0,(m)}_{T},Y^{NV,0,(m)}_T)+\frac{1}{M_2}\sum_{m=1}^{M_2} n\left(f(\hat{X}^{NV,1,(m)}_{T},\hat{Y}^{NV,1,(m)}_{T})-f(\hat{X}^{NV,0,(m)}_{T},Y^{NV,0,(m)}_T)\right).$$
\red{The number of samples $M_1$ and $M_2$ are chosen in order to minimize the computational cost for a given statistical precision $\varepsilon>0$, using the same formulas~\eqref{def_M1} and~\eqref{def_M2} with $(\hat{X}^{NV},\hat{Y}^{NV})$ instead of $(\hat{X}^{Ex},Y)$.} Again, to reduce the variance of the estimator, it is important to use the same noise for the coarse and the refined grids. In particular, we take for the scheme $(\hat{X}^{NV,0}_{kh_1},\hat{Y}^{NV,0}_{kh_1})$ on the coarse grid 
$$ N_{\kappa+1}= N^{\textup{st}} \text{ and }G_{\kappa+1}= G^{\textup{st}}:= \frac{1}{\sqrt{n}} \sum_{k=1}^n \tilde{G}_k.$$
Another choice will be considered for $N_{\kappa+1}$ in Subsection~\ref{Subsec_coupling}, but we will always use $G_{\kappa+1}= G^{\textup{st}}_{\kappa+1}$ in our experiments.

\begin{figure}[h!]
  \centering
  \begin{subfigure}[h]{0.49\textwidth}
    \centering
    \includegraphics[width=\textwidth]{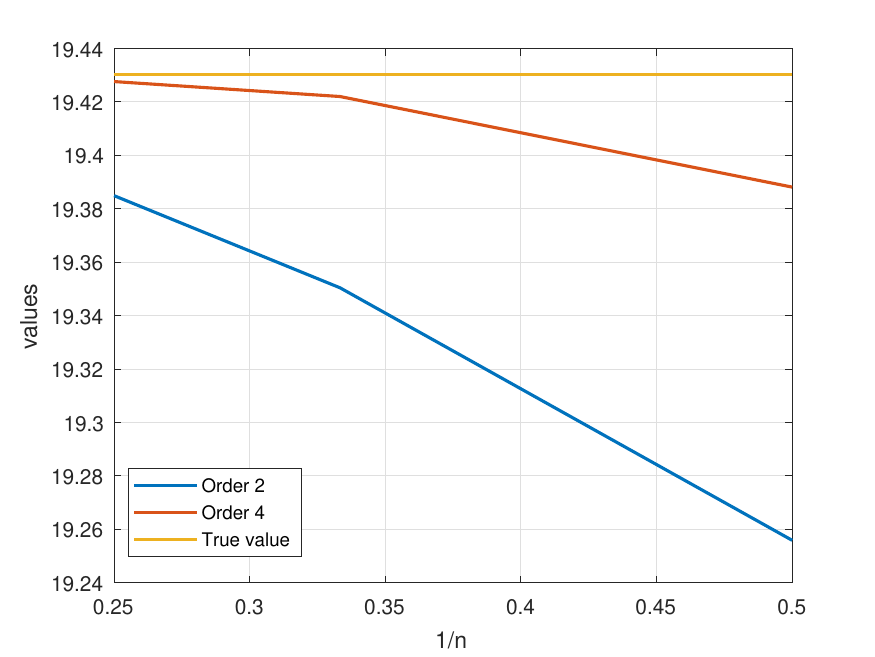}
    \caption{Values plot}
    \label{fig:values_plot_heston1}
  \end{subfigure}
  \hfill
  \begin{subfigure}[h]{0.49\textwidth}
    \centering
    \includegraphics[width=\textwidth]{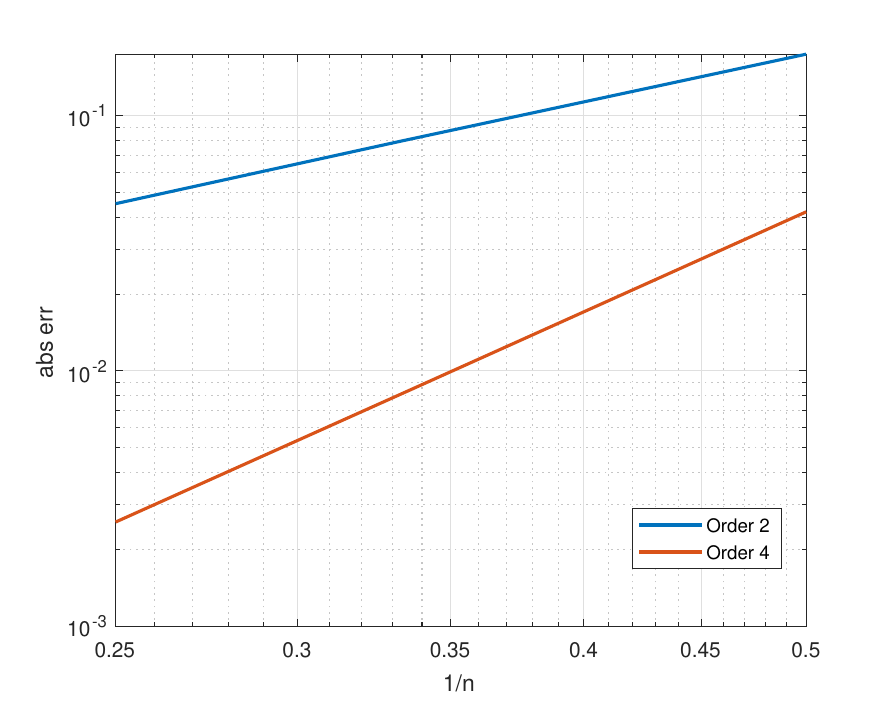}
    \caption{Log-log plot}
    \label{fig:log-log_plot_heston1}
  \end{subfigure}
  \caption{Test function: $f(x,y)=(K-e^x)^+$. Parameters: $S_0=e^{x}=100$, $r=0$, $y=0.2$, $a=0.2$, $b=1$, $\sigma=0.5$, $\rho=-0.7$, $T=1$, $K=105$. Statistical precision $\varepsilon=5$e-4 \red{(too small to be visible on the plots)}.
  Graphic~({\sc a}) shows the Monte Carlo estimated values of $\cPh^{NV,1,n}f$, $\cPh^{NV,2,n}f$ as a function of the time step $1/n$  and the exact value. \red{Graphic~({\sc b}) draws $|\cPh^{NV,\nu,n}f-P_Tf|$ in function of $1/n$ (in log-log scale): the regressed slopes are 1.89 and 4.27 for the second and fourth order respectively.}}\label{Heston_orders}
\end{figure}

\begin{figure}[h!]
  \centering
  \begin{subfigure}[h]{0.49\textwidth}
    \centering
    \includegraphics[width=\textwidth]{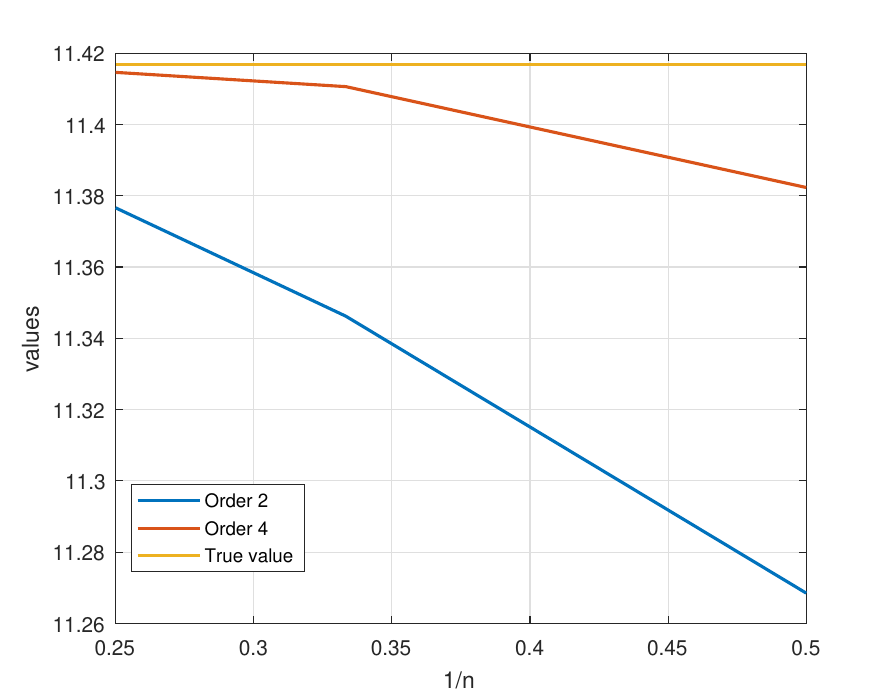}
    \caption{Values plot}
    \label{fig:values_plot_hestonCF1}
  \end{subfigure}
  \hfill
  \begin{subfigure}[h]{0.49\textwidth}
    \centering
    \includegraphics[width=\textwidth]{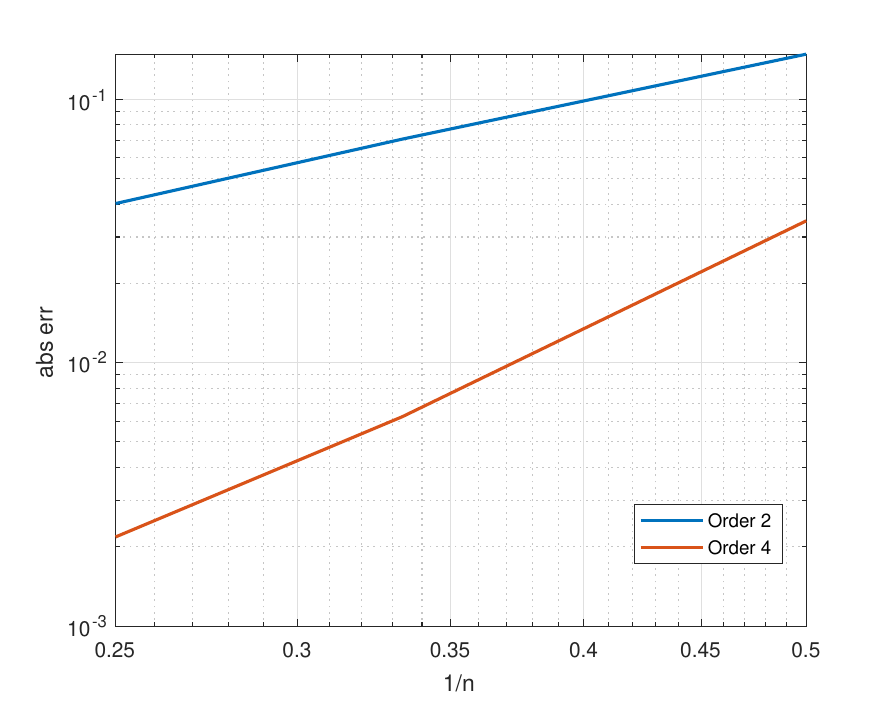}
    \caption{Log-log plot}
    \label{fig:log-log_plot_hestonCF1}
  \end{subfigure}
  \caption{Test function: $f(x,y)=(K-e^x)^+$. Parameters: $S_0=e^x=100$, $r=0$, $y=0.1$, $a=0.1$, $b=1$, $\sigma=1.0$, $\rho=-0.9$, $T=1$, $K=105$. Statistical precision $\varepsilon=5$e-4 \red{(too small to be visible on the plots)}.
  Graphic~({\sc a}) shows the Monte Carlo estimated values of $\cPh^{Ex,1,n}f$, $\cPh^{Ex,2,n}f$ as a function of the time step $1/n$  and the exact value. \red{Graphic~({\sc b}) draws $|\hat{P}^{Ex,\nu,n}f-P_Tf|$ in function of $1/n$ (in log-log scale): the regressed slopes are 1.89 and 4.26 for the second and fourth order respectively.} }\label{HestonCF_orders}
\end{figure}

\subsection{Pricing of European and Asian options}\label{Subsec_Pricing}
We present in Figure~\ref{Heston_orders} the convergence of the approximations $\cPh^{NV,1,n}$ and $\cPh^{NV,2,n}$ for the price of a European option in a case where $\sigma^2\le 4a$. On the left graphic, we draw the values in function of the time step and the exact value of the option price $P_Tf$, that can be calculated with Fourier transform techniques. On the right graphic is plotted the log error in function of the log time step: the estimated slopes are in line with the theoretical order of convergence (2 and 4), even though the test function $f(x)=(K-e^x)_+$ is not as regular as required by Theorem~\ref{main_theorem}.  In Figure~\ref{HestonCF_orders}, we illustrate similarly the convergence of the approximations $\cPh^{Ex,1,n}$ and $\cPh^{Ex,2,n}$ for the price of a European option in a case where $\sigma^2\gg4a$. Again, we observe the theoretical rates of convergence given by Theorem~\ref{main_theorem}.

We now consider the case of Asian options, for which we need to simulate a third coordinate: $\mathcal{I}_t=\int_0^t S^{s,y}_u du=\int_0^t e^{X^{x,y}_u} du$. We explain how to simulate this coordinate for $\hat{P}^{Ex}$, and we do exactly the same then for $\hat{P}^{NV}$. We approximate the integral $\mathcal{I}_t$ by the trapezoidal rule. This gives 
\begin{align*}
  &\hat{\cI}^{Ex,0}_{k h_1} =  \red{\hat{\cI}^{Ex,0}_{(k-1) h_1} +} \frac{e^{\hat{X}^{Ex,0}_{(k-1) h_1}}+e^{\hat{X}^{Ex,0}_{k h_1}}}2 h_1 , \ 1\le k \le n, \\
  &\hat{\cI}^{Ex,1}_{k h_1}=\hat{\cI}^{Ex,0}_{k h_1}, \ 0\le k \le \kappa, \\
  &\hat{\cI}^{Ex,1}_{\kappa h_1 +k'h_2}= \red{\hat{\cI}^{Ex,1}_{\kappa h_1 +(k'-1) h_2} +}\frac{e^{\hat{X}^{Ex,1}_{\kappa h_1+(k'-1)h_2}}+e^{\hat{X}^{Ex,1}_{\kappa h_1 +k'h_2}} }2 h_2 , \ 1\le k'\le n, \\
  &\hat{\cI}^{Ex,1}_{k h_1}= \red{\hat{\cI}^{Ex,1}_{(k-1) h_1} +}\frac{e^{\hat{X}^{Ex,1}_{(k-1) h_1 }}+e^{\hat{X}^{Ex,1}_{kh_1}} }2 h_1, \ \kappa+\red{2}< k\le n,
\end{align*}
with $\hat{\cI}^{Ex,0}_{0}=\hat{\cI}^{Ex,1}_{0}=0$.
Let us mention here that the trapezoidal rule corresponds to the Strang splitting for the generator $\mathcal{L}+e^x\partial_{\mathcal{I}}$. Our formalism would allow to analyse the convergence rate for the Strang splitting for $\mathcal{L}+h(x)\partial_{\mathcal{I}}$, when $h$ is smooth with derivatives of polynomial growth. The exponential function does not fit this condition, and we analyse here the convergence on numerical experiments. 

\begin{figure}[h!]
  \centering
  \begin{subfigure}[h]{0.49\textwidth}
    \centering
    \includegraphics[width=\textwidth]{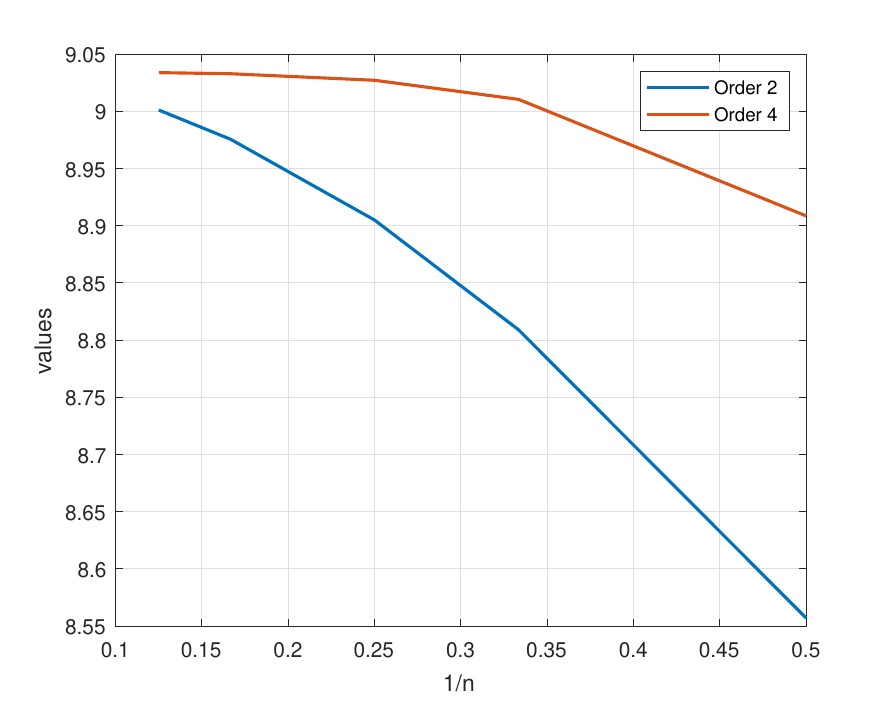}
    \caption{Values plot}
    \label{fig:values_plot_Aheston1}
  \end{subfigure}
  \hfill
  \begin{subfigure}[h]{0.49\textwidth}
    \centering
    \includegraphics[width=\textwidth]{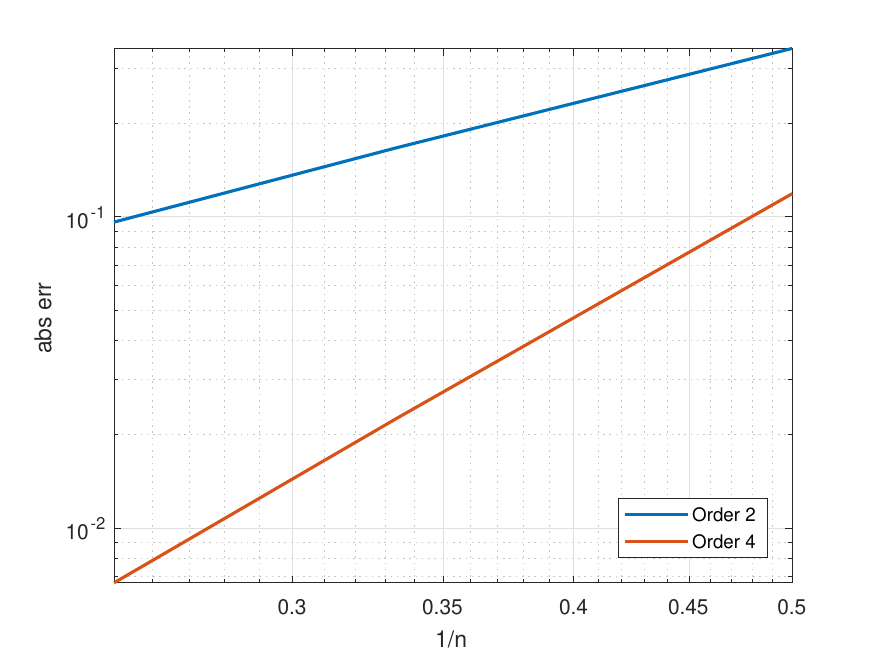}
    \caption{Log-log plot}
    \label{fig:log-log_plot_Aheston1}
  \end{subfigure}
  \caption{Test function: $f(x,y,i)=(K-i/T)^+$. Parameters: $e^{x}=100$, $r=0$, $y=0.2$, $a=0.2$, $b=2$, $\sigma=0.5$, $\rho=-0.7$, $T=1$, $K=100$. Statistical precision $\varepsilon=5$e-4.
  Graphic~({\sc a}) shows the Monte Carlo estimated values of $\cPh^{NV,1,n}f$, $\cPh^{NV,2,n}f$ as a function of the time step $1/n$. \red{Graphic~({\sc b}) draws $|\cPh^{NV,\nu,2n}f-\cPh^{NV,\nu,n}f|$ in function of $1/n$ (in log-log scale): the regressed slopes are 1.85 and 4.30 for the second and fourth order respectively.} }\label{Heston_orders_asian}
\end{figure}

\begin{figure}[h!]
  \centering
  \begin{subfigure}[h]{0.49\textwidth}
    \centering
    \includegraphics[width=\textwidth]{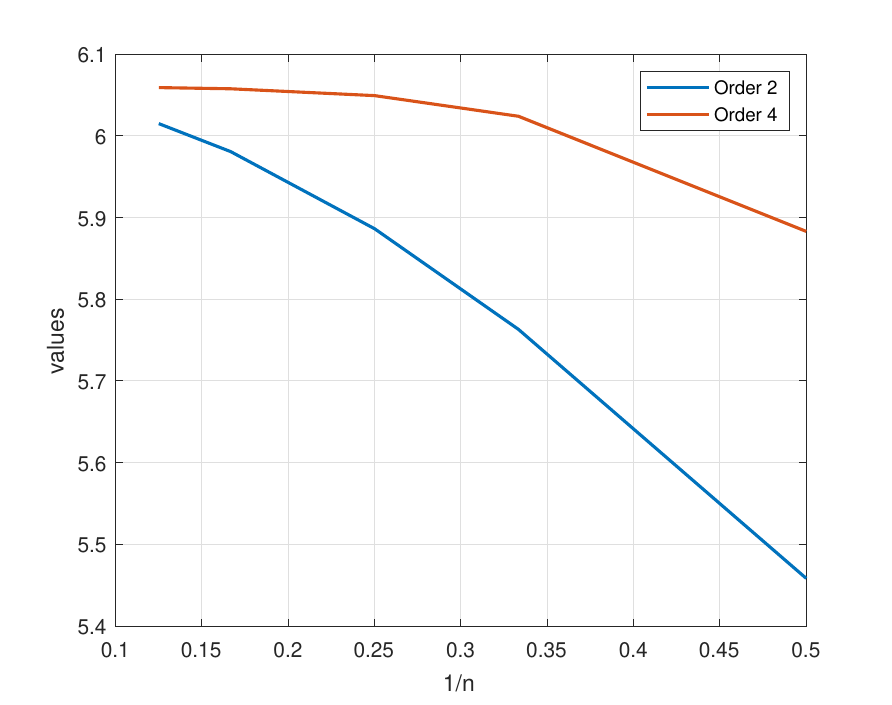}
    \caption{Values plot}
    \label{fig:values_plot_AhestonCF1}
  \end{subfigure}
  \hfill
  \begin{subfigure}[h]{0.49\textwidth}
    \centering
    \includegraphics[width=\textwidth]{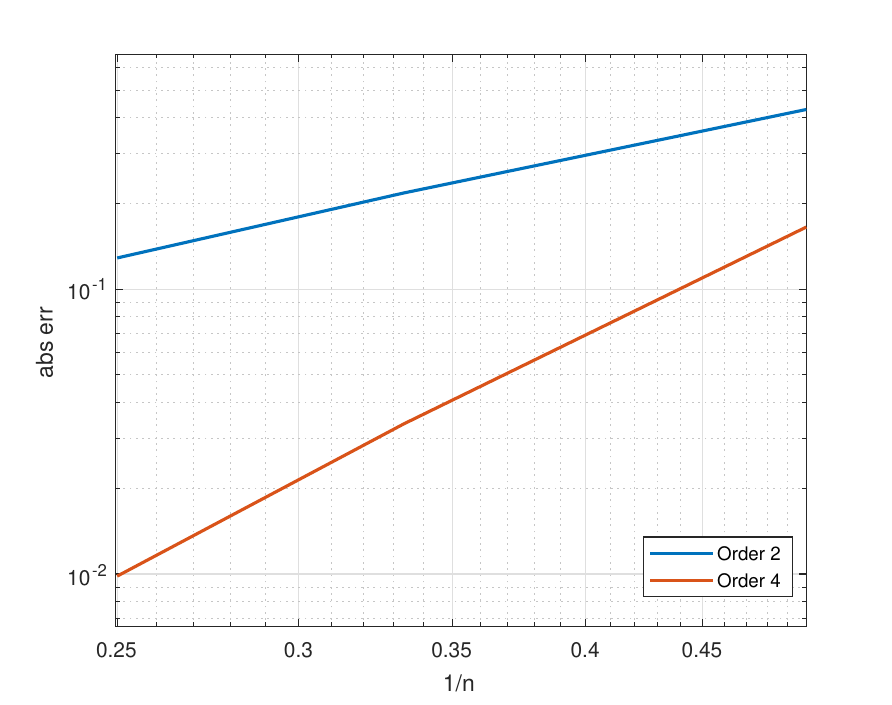}
    \caption{Log-log plot}
    \label{fig:log-log_plot_AhestonCF1}
  \end{subfigure}
  \caption{Test function: $f(x,y,i)=(K-i/T)^+$. Parameters: $e^x=100$, $r=0$, $y=0.1$, $a=0.1$, $b=1$, $\sigma=1.0$, $\rho=-0.9$, $T=1$, $K=100$. Statistical precision $\varepsilon=5$e-4.
  Graphic~({\sc a}) shows the Monte Carlo estimated values of $\cPh^{Ex,1,n}f$, $\cPh^{Ex,2,n}f$ as a function of the time step $1/n$. \red{Graphic~({\sc b}) draws $|\cPh^{Ex,\nu,2n}f-\cPh^{Ex,\nu,n}f|$ in function of $1/n$ (in log-log scale): the regressed slopes are 1.72 and 3.98 for the second and fourth order respectively. }}\label{HestonCF_orders_asian}
\end{figure}

Figure~\ref{Heston_orders_asian} shows the convergence of the approximations~$\cPh^{NV,1,n}$ and~$\cPh^{NV,2,n}$ to calculate the Asian option price $P_Tf=\E[(K-\cI_T\red{/T})^+]$, with $f(x,y,i)=(K-i\red{/T})^+$.  The left graphic draws the obtained value in function of the time step. This time, we do not have an exact value, and we draw in the log-log plot the logarithm of the difference between $\cPh^{NV,\nu,2n}$ and $\cPh^{NV,\nu,n}$. If  $\cPh^{NV,\nu,n}=P_Tf +\frac{c}{n^\eta} +o(n^{-\eta})$ for some $\eta>0$, then $\log(|\cPh^{NV,\nu,2n}-\cPh^{NV,\nu,n}|)=\log(|c|(1-2^{-\eta}))-\eta \log(n) +o(\log(n))$, and therefore the slope of the log-log plot can be seen as an estimation of the rate of convergence. Again, we find empirical rates that are close to 2 for $\nu=1$ and 4 for $\nu=2$, which is in line with the theoretical results. The same observation holds in Figure~\ref{HestonCF_orders_asian} for $\cPh^{Ex,\nu,n}$ in a case where $\sigma^2\ge 4a$.

\red{Last, we compare the second order and fourth order approximations regarding the computation time. We have plotted in Figure~\ref{Heston_4vs2_time} the computation time needed to calculate the Monte Carlo estimator of the second order scheme and of the fourth order approximation for a given statistical precision~$\varepsilon=1e-3$. We have used the one-step coupling presented afterwards in Subsection~\ref{Subsec_coupling} between the coarse and fine time grids. The blue circles (resp. orange squares) indicate the obtained values with the second (resp. fourth) order approximation using $n^2$ time steps (resp. $n$ time steps, one of which is refined $n$ times). We note that their bias are indeed of the same magnitude, but the time required by the fourth order scheme is much lower. In both cases, the computation time is above three times lower for $n=5$, which shows the interest of using random grids.}
\begin{figure}[h!]
  \centering
  \begin{subfigure}[h]{0.49\textwidth}
    \centering
    \includegraphics[width=\textwidth]{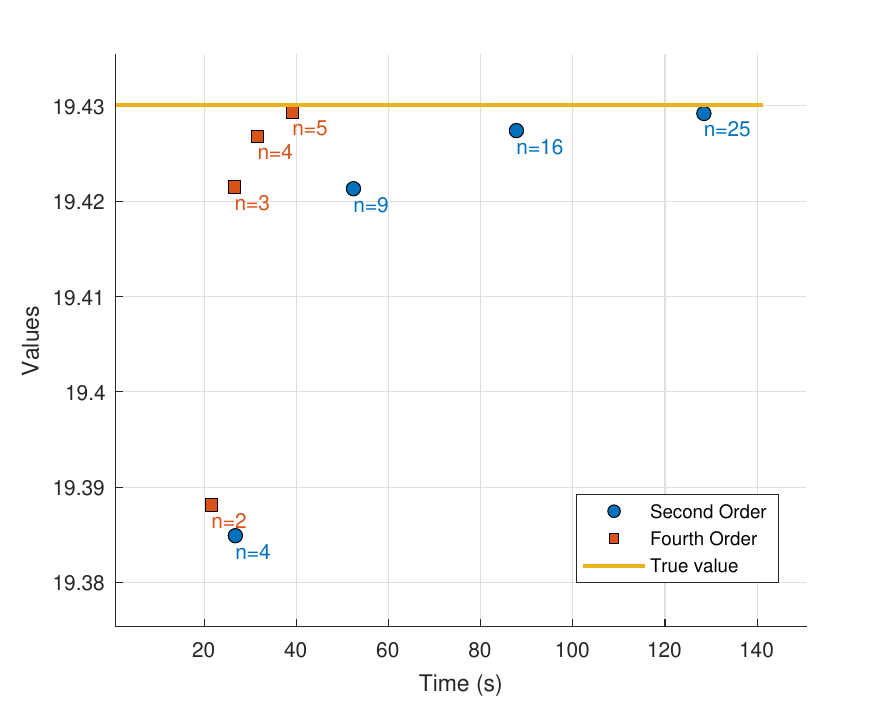}
    \caption{$\sigma^2<4a$}
    \label{fig:values_plot_heston1}
  \end{subfigure}
  \hfill
  \begin{subfigure}[h]{0.49\textwidth}
    \centering
    \includegraphics[width=\textwidth]{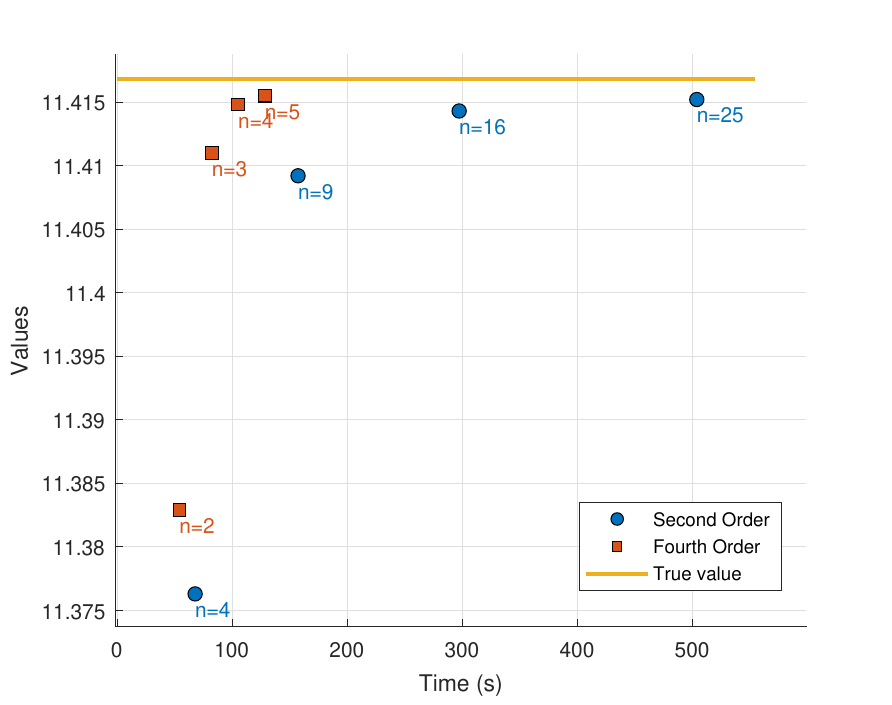}
    \caption{$\sigma^2>4a$}
    \label{fig:log-log_plot_heston1}
  \end{subfigure}
  \caption{Test function: $f(x,y)=(K-e^x)^+$. Statistical precision $\varepsilon=1$e-3.\\
  Graphic~({\sc a}) shows the Monte Carlo estimated values of $\cPh^{NV,1,n^2}f$, $\cPh^{NV,2,n}f$ as a function of the execution time. Parameters: $S_0=e^{x}=100$, $r=0$, $y=0.2$, $a=0.2$, $b=1$, $\sigma=0.5$, $\rho=-0.7$, $T=1$, $K=105$.\\
  Graphic~({\sc b}) shows the Monte Carlo estimated values of $\cPh^{Ex,1,n^2}f$, $\cPh^{Ex,2,n}f$ as a function of the execution time. Parameters: $S_0=e^{x}=100$, $r=0$, $y=0.1$, $a=0.1$, $b=1$, $\sigma=1.0$, $\rho=-0.9$, $T=1$, $K=105$. }\label{Heston_4vs2_time}
\end{figure}

\subsection{Estimators variance and schemes coupling}\label{Subsec_coupling}

In this paragraph,  we discuss  how to couple the refined path and the coarse one in order to minimize the variance of the correction term $$n\left(f(\hat{X}^{SCH,1}_T,\hat{Y}^{SCH,1}_T)-f(\hat{X}^{SCH,0}_T,\hat{Y}^{SCH,0}_T)\right),$$
where $SCH\in \{Ex, NV\}$ indicates the scheme used. We will note $$\mathbf{V}(n)=\text{Var}\left(n\left(f(\hat{X}^{SCH,1}_T,\hat{Y}^{SCH,1}_T)-f(\hat{X}^{SCH,0}_T,\hat{Y}^{SCH,0}_T)\right)\right).$$ 
\red{ Let us note that the variance of $f(\hat{X}^{SCH,0}_T,\hat{Y}^{SCH,0}_T)$ is roughly the one of $f(X_T,Y_T)$ and is thus fixed. In contrast, $\mathbf{V}(n)$ really depends on how the schemes on the coarse and fine grid are coupled. Reducing $\mathbf{V}(n)$ allows to reduce the number of samples required by~\eqref{def_M1} and~\eqref{def_M2}, and thus the computation time.  }

While it is rather natural to take the same driving noise for the other time steps, the difficulty is to find a good coupling on $[\kappa h_1, (\kappa+1) h_1]$ between the noise used on the refined time grid and the one of the coarse grid. This issue does not exist for~$Y$ when it is simulated exactly, and for the Ninomiya-Victoir scheme we always take $G_{\kappa+1}=\frac 1 {\sqrt{n}} \sum_{k=1}^n \tilde{G}_k$. 
We therefore discuss the choice of $N_{\kappa+1}$ that is used for the simulation of~$X$. We consider the two following choices:
$$
N_{\kappa+1}=N^{\textup{st}}=\frac 1{\sqrt{n}}\sum_{k=1}^n\tilde{N}_k,\ \text{ or }  N_{\kappa+1}=N^{\textup{av}} = \frac{\sum_{k=1}^n \sqrt{\hY^{SCH,1}_{\kappa h_1 +(k-1)h_2}+\hY^{SCH,1}_{\kappa h_1 +kh_2}} \tilde{N}_k}{\sqrt{\sum_{k=1}^n \hY^{SCH,1}_{\kappa h_1 +(k-1)h_2}+\hY^{SCH,1}_{\kappa h_1 +kh_2}}}.
$$
Note that $N^{\textup{av}}\sim \mathcal{N}(0,1)$, since the normal variables $\tilde{N}_k$, $1\le k\le n$, are independent of the $Y$ component. This second choice is also rather natural since it weights each normal variable with the corresponding volatility on each fine time-step. A similar coupling has been proposed by Zheng~\cite{ZhengC} in a context of Multi-Level Monte-Carlo for the Heston model.

Besides this choice of coupling, we also consider another scheme for the Heston model. In fact, an alternative of Strang's scheme is to introduce a Bernoulli random variable of parameter~$1/2$ that selects which scheme is used first. We want to see if this additional random variable has an incidence on the variance of the correcting term.  This  scheme is given by
$$\hX^{SCH,x,y}_t = x +(r-\frac{\rho}{\sigma}a)t +\frac{\rho}{\sigma}(\hY^{SCH,y}_t-y) +(\frac{\rho}{\sigma}b-\frac{1}{2})\frac{y+\hY^{SCH,y}_t}{2}t +\sqrt{y + B (1-\rho^2)(\hY^{SCH,y}_t-y) t}N, 
$$
where $N\sim \mathcal{N}(0,1)$ and $B\sim \mathcal{B}(1/2)$ is an independent Bernoulli random variable. The random variable $\hY^{SCH,y}_t$ is either equal to $Y^y_t$ for $SCH=Ex$ or to $\hY^y_t$ for $SCH=NV$. This scheme has been used in the numerical experiments of~\cite{AL} and is indicated with "Bernoulli" in the following tables.

\red{
As far as European options are concerned, we now present another scheme for the distribution of $X$ at time $T$, which as we shall see produces the correction with the lowest variance. It is based on the same argument as the one developed by Romano and Touzi~\cite{RoTo}.  Let  $\hX^{SCH,0}$ and $\hX^{SCH,1}$ where $SCH\in\{NV,Ex\}$ denote the schemes introduced in Subsection~\ref{Subsec_implementation}. We have at the final time $T=nh_1$
\begin{align*}
  \hX^{SCH,0}_{nh_1} =& x +(r-\frac{\rho}{\sigma}a)T +\frac{\rho}{\sigma}(\hY^{SCH,0}_{nh_1}-y) +(\frac{\rho}{\sigma}b-\frac{1}{2}) \sum_{j=0}^{n-1}\frac{h_1}{2}(\hY^{SCH,0}_{jh_1}+\hY^{SCH,0}_{(j+1)h_1}) \\
  &+ \brho \sum_{j=0}^{n-1}\sqrt{\frac{h_1}{2}(\hY^{SCH,0}_{jh_1}+\hY^{SCH,0}_{(j+1)h_1})}N_{j+1}.
\end{align*}
We remark that the law of  $\hX^{SCH,0}_T$ given $(\hY^{SCH,0}_{jh_1}, j\in{0,\ldots,n})$ is normally distributed with mean  and variance respectively equal to
$$
x+(r-\frac{\rho}{\sigma}a)T+\frac{\rho}{\sigma}(\hY^{SCH,0}_{nh_1}-y) +(\frac{\rho}{\sigma}b-\frac{1}{2})\widehat{IY}^{SCH,0}_{nh_1} \text{ and } 
(1-\rho^2)\widehat{IY}^{SCH,0}_{nh_1},
$$
with $$\widehat{IY}^{SCH,0}_{nh_1} = \frac{h_1}{2}\sum_{j=0}^{n-1}(\hY^{SCH,0}_{jh_1}+\hY^{SCH,0}_{(j+1)h_1}).$$
In the same way, the  law of  $\hX^{SCH,1}_T$ given $\hY^{SCH,1}$ is normally distributed with mean  and variance respectively equal to
$$
x+(r-\frac{\rho}{\sigma}a)T+\frac{\rho}{\sigma}(\hY^{SCH,1}_{nh_1}-y) +(\frac{\rho}{\sigma}b-\frac{1}{2})\widehat{IY}^{SCH,1}_{nh_1} \text{ and } 
(1-\rho^2)\widehat{IY}^{SCH,1}_{nh_1},
$$
with $$\widehat{IY}^{SCH,1}_{nh_1} = \frac{h_1}{2}\sum_{j\not= \kappa}(\hY^{SCH,1}_{jh_1}+\hY^{SCH,1}_{(j+1)h_1}) + \frac{h_2}{2}\sum_{j= 0}^{n-1}(\hY^{SCH,1}_{\kappa h_1+jh_2}+\hY^{SCH,1}_{\kappa h_1+(j+1)h_2}).$$
This give us naturally the following simulation scheme
\begin{equation}\label{one_step_coupling}
\hX^{SCH,\ell}_{nh_1} = x +(r-\frac{\rho}{\sigma}a)T +\frac{\rho}{\sigma}(\hY^{SCH,\ell}_T-y) +(\frac{\rho}{\sigma}b-\frac{1}{2}) \widehat{IY}^{SCH,\ell}_{nh_1} + \brho \sqrt{\widehat{IY}^{SCH,\ell}_{nh_1}}N, \ell\in \{0,1\},
\end{equation}
where the same  single standard Gaussian $N$ is used for the coarse and fine grid. 
This scheme does not change the law of $(\hX^{SCH,\ell}_{nh_1},\hY^{SCH,\ell})$ with respect to the scheme presented in Subsection~\ref{Subsec_implementation}, but it allows us to make a “perfect” coupling regarding $Y$-independent Gaussian noise between the refined and coarse grids. We will call this procedure the ``one-step" coupling. Algorithm~\ref{algo} gives the pseudocode to calculate the estimator in Equation~\eqref{estimateur_ex} obtained with this coupling.}

\begin{algorithm}[htp]
  \red{  \caption{Calculation of $\E[f(X_T,Y_T)]$ with the "one-step" coupling, when $Y$ is sampled exactly}\label{algo}
  \begin{algorithmic}
    \STATE Input: $M_1\ge M_2, T, n$, initial values $x\in \R$ and $y\ge 0$, and model parameters
    \STATE $h_1 \leftarrow T/n$, $h_2 \leftarrow T/n^2$
    \STATE $EV\leftarrow 0,\ Corr \leftarrow 0$ \quad // Expected value and correction terms
    \FOR{$m=1$ to $M_1$}
    \IF{$m>M_2$}
    \STATE Draw a path $(Y_{kh_1})_{0\le k\le n}$ with $Y_0=y$ and $N\sim\mathcal{N}(0,1)$
    \STATE $IY\leftarrow \frac{h_1}{2}\sum_{j=0}^{n-1}(Y_{jh_1}+Y_{(j+1)h_1})$
    \STATE $X \leftarrow x+(r-\frac{\rho}{\sigma}a)T+\frac{\rho}{\sigma}(Y_{nh_1}-y) +(\frac{\rho}{\sigma}b-\frac{1}{2})IY+\sqrt{(1-\rho^2)IY}N
    $
    \STATE $EV\leftarrow EV+f(X,Y_{nh_1})$ 
    \ELSE 
    \STATE Draw $\kappa$ uniformly on $\{0,\dots,n-1\}$  and $N\sim\mathcal{N}(0,1)$
    \STATE Draw a path on the fine grid $(Y_{kh_1})_{0\le k\le n}$ and $(Y_{\kappa h_1+kh_2})_{1\le k\le n-1}$ with $Y_0=y$
    \STATE $IY^0\leftarrow \frac{h_1}{2}\sum_{j=0}^{n-1}(Y_{jh_1}+Y_{(j+1)h_1})$
    \STATE $X^0 \leftarrow x+(r-\frac{\rho}{\sigma}a)T+\frac{\rho}{\sigma}(Y_{nh_1}-y) +(\frac{\rho}{\sigma}b-\frac{1}{2})IY^0+\sqrt{(1-\rho^2)IY^0}N
    $
    \STATE $IY^1\leftarrow IY^0- \frac{h_1}{2}(Y_{\kappa h_1}+Y_{(\kappa+1)h_1}) +\frac{h_2}{2}\sum_{j=0}^{n-1}(Y_{\kappa h_1+jh_2}+Y_{\kappa h_1+(j+1)h_2})$
    \STATE $X^1 \leftarrow x+(r-\frac{\rho}{\sigma}a)T+\frac{\rho}{\sigma}(Y_{nh_1}-y) +(\frac{\rho}{\sigma}b-\frac{1}{2})IY^1+\sqrt{(1-\rho^2)IY^1}N
    $
    \STATE $EV\leftarrow EV+f(X^0,Y_{nh_1})$
    \STATE $Corr\leftarrow Corr+n(f(X^1,Y_{nh_1})-f(X^0,Y_{nh_1}))$
    \ENDIF
    \ENDFOR
    \RETURN $\frac{EV}{M_1}+\frac{Corr}{M_2}$
  \end{algorithmic}}
\end{algorithm}

\begin{table}[h!]
  \begin{tabular}{ |c |c|c|c|c|c|c|  }
    \cline{1-7}
        Scheme & Coupling & $n=2$     & $n=4$      & $n=8$       & $n=16$    & $n=32$    \\
      \hline
        \multicolumn{1}{|c|}{\multirow{2}{*}{$NV$}} &  \multirow{2}{*}{$N^{\textup{st}}$}    & 12.13 & 18.48 & 21.85 & 23.56 & 24.41 \\ 
        \multicolumn{1}{|c|}{}&     & (0.01) & (0.01) & (0.01) & (0.02) & (0.02) \\
      \hline
        \multicolumn{1}{|c|}{ \multirow{2}{*}{$NV$}} &  \multirow{2}{*}{$N^{\textup{av}}$}      & 8.31 & 9.08 & 8.91 & 8.70 & 8.57\\ 
        \multicolumn{1}{|c|}{}& & (0.01) & (0.01) & (0.01) & (0.01) & (0.01) \\ 
      \hline
        \multicolumn{1}{|c|}{ \multirow{2}{*}{$\red{NV}$}} &  \multirow{2}{*}{$\red{\text{one-step}}$}      & 5.58 & 3.52 & 1.88 & 0.97 & 0.49\\ 
        \multicolumn{1}{|c|}{}& & (0.005) & (0.004) & (0.003) & (0.002) & (0.001) \\ 
      \hline
        \multicolumn{1}{|c|}{\multirow{2}{*}{$NV$, Bernoulli}} &  \multirow{2}{*}{$N^{\textup{st}}$} & 33.27 & 41.96 & 46.14 & 48.27 & 49.37 \\ 
        \multicolumn{1}{|c|}{}&      & (0.02) & (0.03) & (0.03) & (0.04) & (0.04) \\
      \hline
        \multicolumn{1}{|c|}{\multirow{2}{*}{$NV$, Bernoulli}} &  \multirow{2}{*}{$N^{\textup{av}}$} & 25.11 & 28.55 & 30.74 & 32.13 & 32.85\\ 
        \multicolumn{1}{|c|}{}&   & (0.02) & (0.02) & (0.03) & (0.03) & (0.03) \\       
    \hline
      \multicolumn{1}{|c|}{\multirow{2}{*}{$Ex$}} &  \multirow{2}{*}{$N^{\textup{st}}$}     & 30.19 & 30.19 & 28.09 & 26.74 & 26.02 \\ 
      \multicolumn{1}{|c|}{}&  & (0.02) & (0.02) & (0.02) & (0.02) & (0.02) \\
    \hline
      \multicolumn{1}{|c|}{\multirow{2}{*}{$Ex$}} &  \multirow{2}{*}{$N^{\textup{av}}$}  & 26.35 & 20.80 & 15.17 & 11.88 &  10.18\\ 
      \multicolumn{1}{|c|}{}& & (0.01) & (0.01) & (0.01) & (0.01) & (0.01) \\ 
    \hline
      \multicolumn{1}{|c|}{\multirow{2}{*}{$\red{Ex}$}} &  \multirow{2}{*}{$\red{\text{one-step}}$}  & 23.58 & 15.15 & 8.04 & 4.08 &  2.05\\ 
      \multicolumn{1}{|c|}{}& & (0.011) & (0.007) & (0.004) & (0.002) & (0.001) \\
    \hline
    
  \end{tabular}
  \caption{Variance $\mathbf{V}(n)$ estimated  with $10^8$ samples, the 95\% confidence precision is indicated below in parentheses. Test function: $f(x,y)=(K-e^x)^+$. Parameters: $e^x=100$, $r=0$, $x=0.2$, $a=0.2$, $b=1.0$, $\sigma=0.5$, $\rho=-0.7$, $T=1$, $K=105$.}\label{Table_VAR_Heston}
\end{table}

\begin{table}[h!]
  \begin{tabular}{ |c |c|c|c|c|c|c|  }
    %\hline
    %\multicolumn{5}{|c|}{Execution Time}            \\
    \cline{1-7}
      Scheme & Coupling & $n=2$     & $n=4$      & $n=8$       & $n=16$    & $n=32$    \\
    \hline
      \multicolumn{1}{|c|}{\multirow{2}{*}{$Ex$}} &  \multirow{2}{*}{$N^{\textup{st}}$}      & 38.69 & 39.51 & 36.96 & 35.23 & 34.32 \\ 
      \multicolumn{1}{|c|}{}&      & (0.03) & (0.03) & (0.03) & (0.03) & (0.03) \\
    \hline
      \multicolumn{1}{|c|}{\multirow{2}{*}{$Ex$}} &  \multirow{2}{*}{$N^{\textup{av}}$}     & 32.49 & 26.01 & 19.16 & 15.20 & 13.17 \\ 
      \multicolumn{1}{|c|}{}& & (0.02) & (0.02) & (0.02) & (0.01) & (0.01) \\
    \hline
      \multicolumn{1}{|c|}{\multirow{2}{*}{$\red{Ex}$}} &  \multirow{2}{*}{$\red{\text{one-step}}$}     & 29.60& 19.30 & 10.29 & 5.24 & 2.63 \\ 
      \multicolumn{1}{|c|}{}& & (0.02) & (0.013) & (0.008) & (0.004) & (0.002) \\ 
    \hline
        \multicolumn{1}{|c|}{\multirow{2}{*}{$Ex$, Bernoulli}} &  \multirow{2}{*}{$N^{\textup{st}}$}  & 65.66 & 68.93 & 66.95 & 65.47 & 65.01 \\ 
        \multicolumn{1}{|c|}{}&      & (0.04) & (0.05) & (0.05) & (0.05) & (0.05) \\
      \hline
        \multicolumn{1}{|c|}{\multirow{2}{*}{$Ex$, Bernoulli}} &  \multirow{2}{*}{$N^{\textup{av}}$}    & 61.04 & 57.45 & 50.98 & 47.03 & 45.12 \\ 
        \multicolumn{1}{|c|}{}&   & (0.04) & (0.04) & (0.04) & (0.04) & (0.04) \\      
    \hline
    
  \end{tabular}
  \caption{Variance $\mathbf{V}(n)$ estimated  with $10^8$ samples,  the 95\% confidence precision is indicated below in parentheses. Test function: $f(x,y)=(K-e^x)^+$. Parameters: $e^x=100$, $r=0$, $x=0.1$, $a=0.1$, $b=1.0$, $\sigma=1.0$, $\rho=-0.9$, $T=1$, $K=105$.}\label{Table_VAR_HestonCF}
\end{table}

We have reported in Tables~\ref{Table_VAR_Heston} and \ref{Table_VAR_HestonCF} the variance of the correcting term for the different schemes, the two different choices for $N_{\kappa+1}$, \red{the one-step coupling}, and different values of~$n$.  Table~\ref{Table_VAR_Heston} reports a case with $\sigma^2\le 4a$ where the Ninomiya-Victoir scheme is well defined, while Table~\ref{Table_VAR_HestonCF} reports a case with $\sigma^2>4a$. In both cases, we have taken the example of a European Put option. In both tables, we observe that the scheme using a Bernoulli random variable has a correcting term of  higher variance. Besides, it requires to simulate one more random variable. Thus, the schemes based on the Strang composition are better suited with the convergence acceleration using random grids.   

We now comment the coupling of the schemes. In all our experiments, the coupling using~$N^{\textup{av}}$ gives a lower variance than the one using~$N^{\textup{st}}$. Besides, we observe that the gain factor between the two choices is  increasing with~$n$. We have a gain factor of $\frac{24.41}{8.57}\approx 2.85$ in Table~\ref{Table_VAR_Heston} for the Ninomiya-Victoir scheme and $n=32$, and of $2.32$ in Table~\ref{Table_VAR_HestonCF} for the scheme $Ex$ with $n=16$. As a consequence, we recommend the use of $N^{\textup{av}}$ to couple the schemes on the coarse and fine grids.  \red{The ``one-step" coupling is even much better. It gives a variance that goes to~$0$ as $n\to \infty$ and seems to scale linearly with $1/n$. We recall however that this coupling can only be used for European options, but we have to use the other ones for pathwise options.}

\subsection{Towards higher order approximations of Rough Heston process}
In this last paragraph, we propose to investigate numerically the approximations with random grids in the case of the rough Heston model. We first recall that the rough Heston model proposed by El Euch and Rosenbaum~\cite{EER19} is given by $S_t=e^{X^{x,y}_t}$, where
\begin{align}
  X^{x,y}_t &= x + \int_0^t \left(r-\frac 12 Y^y_u \right) du + \int_0^t\sqrt{Y^y_u}  (\rho dW_u + \sqrt{1-\rho^2} dB_u),\\
  Y^y_t &= y + \int_0^t K(t-u) (a-bY^y_u) du +\int_0^t K(t-u)\sigma\sqrt{Y^y_u} dW_u, \label{rough_vol_SVE}
\end{align}
 where $K$ is the fractional kernel given by
\begin{equation}\label{fractional_kernel}
  K(t)= \frac{t^{H-1/2}}{\Gamma(H+1/2)}
\end{equation}
with Hurst parameter $H\in(0,1/2)$.
The convolution through the kernel $K$ in \eqref{rough_vol_SVE} introduces a dependence of the volatility $Y$ on the past, and the process $(X,Y)$ is not Markovian. Despite this, it is possible to find a process in larger dimension that is Markovian and approximates the rough process well. It is well known (see e.g. Alfonsi and Kebaier \cite[Proposition 2.1]{AK24}) that if we replace the rough kernel $K$ in \eqref{rough_vol_SVE}  by a discrete completely monotone kernel
\begin{equation}\label{discrete_completely_monotone_kernel}
  \tK(t) = \sum_{k=1}^d \gamma_k e^{-\rho_k t},\qquad \gamma_k,\rho_k\ge 0,\,k\in\{1,\ldots,d\},
\end{equation} 
then the solution of the Stochastic Volterra Equation
\begin{equation}\label{dcmk_vol_SVE}
  \tY_t = y + \int_0^t \tK(t-u) (a-b\tY^y_u) du +\int_0^t \tK(t-u)\sigma\sqrt{\tY^y_u} dW_u, 
\end{equation}
is given by $\tY_t = y + \sum_{k=1}^d \gamma_k\tY^k_t$, where $\tbfY=(\tY^1,\ldots,\tY^d)$ solves the SDE in $\R^d$:
\begin{equation}\label{dcmk_vol_SDE}
  \tY^k_t = -\rho_k\int_0^t \tY^k_u du + \int_0^t (a-b\tY_u) du + \int_0^t \sigma \sqrt{\tY_u}dW_u,\quad k\in\{1,\ldots,d\}, t\ge 0.
\end{equation}
We want to build a second order scheme for \eqref{dcmk_vol_SDE} along with
\begin{equation*}
  \tX_t = x + \int_0^t(r-\frac 12 \tY_u) du + \int_0^t\sqrt{\tY_u} (\rho dW_u + \sqrt{1-\rho^2} dB_u).
\end{equation*}
This multifactor model has been first developed by Abi Jaber and El Euch \cite{AEE19} and can be seen under a suitable choice of $\tK(t) = \sum_{k=1}^d \gamma_k e^{-\rho_k t}$  as an approximation of the rough Heston model.

We present here a second order approximation scheme for the couple $(\tX,\tY)$ that preserve the positivity of $\tY$ as proved by Alfonsi in \cite[Theorem 4.2 and Subsection 4.3]{AA_NPK}. The infinitesimal generator of the $d+1$ dimensional process $(\tX,\tbfY)$ is given by 
\begin{multline}\label{inf_gen_dcmk_vol}
  \mcL f(x,\bfy)  = (r-\frac{1}{2}y')\partial_{x} f(x,\bfy)+\sum_{k=1}^d (a-\rho_k y_k-b y')\partial_{y_k} f(x,\bfy)\\ 
  +\frac{1}{2} \partial_x^2f(x,\bfy)+\sum_{k=1}^d 2\rho\sigma \partial_{x}\partial_{y_k} f(x,\bfy) +\frac{1}{2} \sum_{k,l=1}^d \sigma^2 y' \partial_{y_k}\partial_{y_l} f(x,\bfy),
\end{multline}
where $\bfy = (y_1,\ldots,y_d)$ and $y'=y+\sum_{j=1}^d \gamma_j y_j$. We use the following splitting $\mcL=\mcL_1+\mcL_2$, where $\mcL_1 f = -\sum_{k=1}^d \rho_k y_k \partial_{y_k} f$ is the infinitesimal generator associated to
\begin{equation}\label{mcL1_linear_ODE}
  \begin{aligned}
    dX_t   &=0, \\
    dY^k_t &= -\rho_k Y^k_t dt,\qquad k\in\{1,\ldots,d\},
  \end{aligned}
\end{equation}
and $\mcL_2$ is associated to
\begin{equation}\label{mcL2_SDE}
  \begin{aligned}
    dX_t   &=(r-\frac{1}{2}Y_t)d_t+\sqrt{Y_t} (\rho dW_t + \sqrt{1-\rho^2} dB_t), \\
    dY^k_t &=  (a-bY_t) dt +  \sigma \sqrt{Y_t}dW_t, \text{ with }Y_t = y+\sum_{k=1}^d \gamma_k Y^k_t \quad k\in\{1,\ldots,d\}.
  \end{aligned}
\end{equation}
The linear ODE \eqref{mcL1_linear_ODE} has the exact solution
\begin{equation}
  \psi_1(t,x,{\bf y})=(x,{\bf y}_t), \text { with } {\bf y}_t=(y_1e^{-\rho_1 t}, \ldots, y_de^{-\rho_d t}).
\end{equation}
From \eqref{mcL2_SDE}, we obtain that $(X_t,Y_t)$ satisfies the following log-Heston SDEs
\begin{equation}\label{K0_logHeston_SDE}
  \begin{aligned}
    X_t   &= x + \int_0^t(r-\frac{1}{2}Y_u)d_t+\int_0^t\sqrt{Y_u} (\rho dW_u + \sqrt{1-\rho^2} dB_u), \\
    Y_t &= y' + \int_0^t K(0)(a-bY_u) du +\int_0^tK(0)\sigma \sqrt{Y_u} dW_u,
  \end{aligned}
\end{equation}
and $dY^k_t=\frac{1}{K(0)} dY_t$ (note that $K(0)=\sum_{j=1}^d \gamma_j$).
So, having a second order scheme $(\hat{X}^{x,y'}_t,\hat{Y}^{y'}_t)$ for $(X_t,Y_t)$, we can build a second order scheme for~\eqref{mcL2_SDE} by
\begin{equation}
  (\hat{X}^{x,y}_t,\hat{Y}^{1,y}_t,\dots,\hat{Y}^{1,y}_t) = (\hat{X}^{x,y'}_t,A_\bfy(\hat{Y}^{y'}_t)),
\end{equation}
where 
\begin{equation}
  A_\bfy(z) = \left(y_1 +\frac{z-y'}{K(0)}, \ldots, y_d+\frac{z-y'}{K(0)}  \right).
\end{equation}
In the end, we use again the Strang composition to get the second order scheme for~\eqref{inf_gen_dcmk_vol} starting from $(x,{\bf y})$ and time-step $t>0$:
\begin{equation}\label{scheme_multif_Heston} \psi_1 \left(t/2,\hat{X}^{x,y'_{t/2}}_t,A_{\bfy_{t/2}}(\hat{Y}^{y'_{t/2}}_t) \right),\end{equation}
where $y'_{t/2}=y+\sum_{j=1}^d \gamma_j y_j e^{-\rho_jt/2}$.

Now that we have defined the approximation scheme~\eqref{scheme_multif_Heston} for the multifactor Heston model, we want to use it to test numerically the convergence acceleration provided by the random grids. The construction of the estimators is identical to the one of $\cPh^{NV,1,n}$ and $\cPh^{NV,2,n}$ in Subsection~\ref{Subsec_implementation} and we do not reproduce it here. Also, by a slight abuse of notation, we still denote by $\cPh^{NV,1,n}$ and $\cPh^{NV,2,n}$ these estimators that are well defined $\tK(0)\sigma^2<4a$. Unfortunately, there does not exist yet -- up to our knowledge -- efficient exact simulation method for the multifactor Cox-Ingersoll-Ross process. It it were the case, we could then define the corresponding estimators $\cPh^{Ex,1,n}$ and $\cPh^{Ex,2,n}$ exactly as in Subsection~\ref{Subsec_implementation}, for any $\sigma>0$. Here, we thus present only simulation in the case $\tK(0)\sigma^2<4a$. These simulations are intended to be a first attempt to get higher order approximations of the multifactor Heston model. We let the case $\tK(0)\sigma^2>4a$ as well as theoretical proofs of convergence in this model for future studies.

Multi exponential approximations of the rough kernel are available in literature, see e.g. Abi Jaber, El Euch \cite{AEE19} and Alfonsi, Kebaier \cite{AK24}. In our simulation we will use the algorithm BL2 suggested by Bayer and Breneis in \cite{BB23}, that optimizes the $\L^2([0,T])$-error between $K$ and $\tK$ while limiting high values of $\rho_k$. In particular, we will use the approximate BL2 Kernel with $d=3$ exponential factors, that has been proven to approximate a whole volatility surface of rough Heston call prices with approximately 1\% of maximal relative error~\cite[Table 4, third column]{BB23}.
When the Hurst parameter $H=0.1$ the nodes and weights are resumed following table
\begin{center}
  \begin{tabular}{ |r|r|r| } 
   \hline
   $\rho_1 = 0.08399474$ & $\rho_2 = 5.64850577$ & $\rho_3 = 118.00624702$ \\
   \hline
   $\gamma_1 = 0.80386099$ & $\gamma_2 = 1.60786461$ & $\gamma_3 = 8.80775525$  \\ 
   \hline
  \end{tabular}
  \end{center}
We consider European put option prices and present in Figure~\ref{fig:values_plot_rheston1}  a plot of the values of $\cPh^{NV,1,n}f$ and $\cPh^{NV,2,n}f$ as a function of the time step with the exact value obtained by Fourier techniques. In Figure~\ref{fig:log-log_plot_rheston1}, we draw a log-log plot to quantify the order of convergence.
First, we observe that we obtain a much larger bias than in our previous numerical experiments for the Heston process, Figure~\ref{fig:values_plot_hestonCF1}. This is mainly due to the map $\psi_1$ that has relatively large nodes, namely $\rho_2$ and especially $\rho_3$. The contribution of these exponential factors in the dynamics of the scheme gets more important when the time step is sufficiently small. However, even if the bias is more important, the speed of convergence are still in line with the theoretical ones.  The regressed slopes for $\cPh^{NV,1,n}f$ and $\cPh^{NV,2,n}f$ are respectively 1.89 and 3.98, showing that the scheme is indeed a second-order scheme and that the boosting technique with random grids works again in this case.

\begin{figure}[h!]
  \centering
  \begin{subfigure}[h]{0.49\textwidth}
    \centering
    \includegraphics[width=\textwidth]{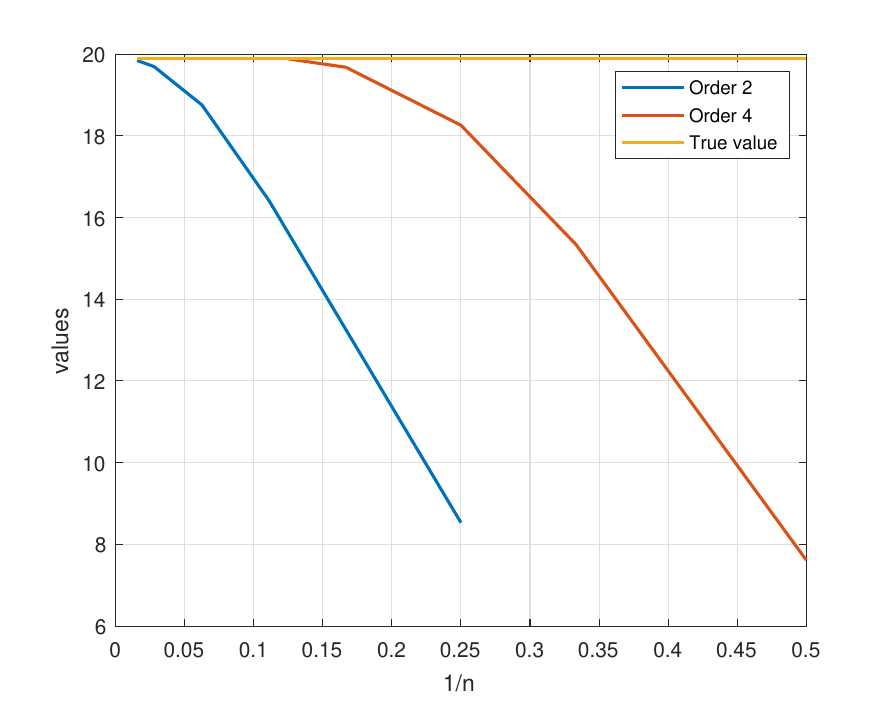}
    \caption{Values plot}
    \label{fig:values_plot_rheston1}
  \end{subfigure}
  \hfill
  \begin{subfigure}[h]{0.49\textwidth}
    \centering
    \includegraphics[width=\textwidth]{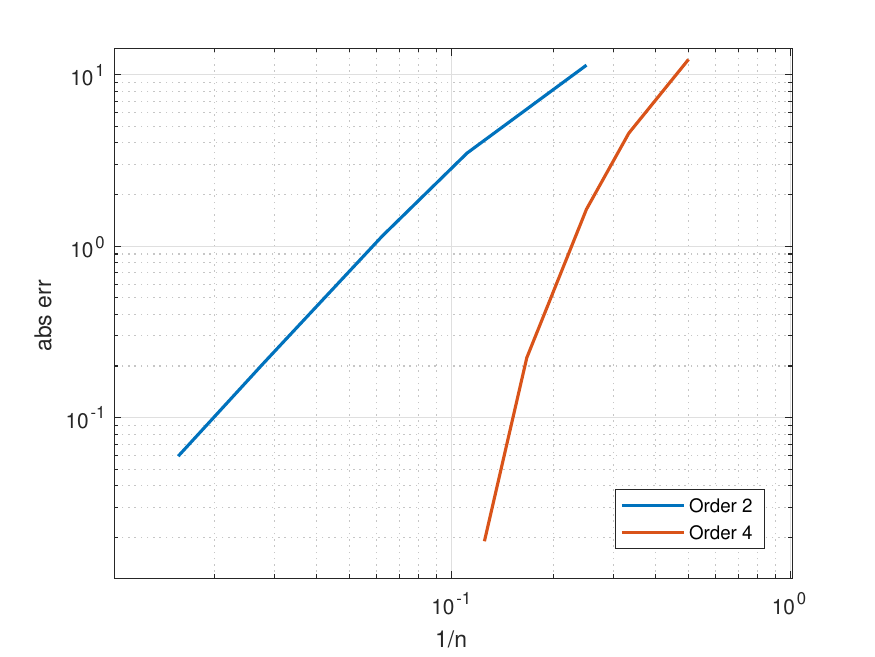}
    \caption{Log-log plot}
    \label{fig:log-log_plot_rheston1}
  \end{subfigure}
  \caption{Test function: $f(x,y)=(K-e^x)^+$. Parameters: $S_0=e^x=100$, $r=0$, $y=0.1$, $a=0.3$, $b=1$, $\sigma=0.1$, $\rho=-0.7$, $T=1$, $K=105$. Statistical precision $\varepsilon=2$e-3.
  Graphic~({\sc a}) shows the Monte Carlo estimated values of $\cPh^{NV,1,n}f$, $\cPh^{NV,2,n}f$ as a function of the time step $1/n$  and the exact option value. \red{Graphic~({\sc b}) draws $|\cPh^{NV,\nu,n}f-P_Tf|$ in function of $1/n$ (in log-log scale): the regressed slopes are 1.89 and 3.98 for the second and fourth order respectively.}}\label{rHeston_orders}
\end{figure}

\red{
\subsection{Further applications to financial models}
This subsection explores briefly straightforward generalizations of the approximation scheme developed in the article, leveraging the existing theoretical framework. We consider the following general dynamics:
\begin{equation}\label{MF_nocorr_log-Heston_SDEs} 
  \begin{cases}
    dX_t &= (r_t-\frac{1}{2}\sum_{m=1}^{M}Y_{m,t}) dt + \sum_{m=1}^{M}\sqrt{Y_{m,t}} (\rho_m dW_{m,t} + \sqrt{1-\rho_m^2} dB_{m,t})+dH_t, \\
    dY_{m,t} &= (a_m-b_mY_{m,t}) dt +\sigma_m \sqrt{Y_{m,t}} dW_{m,t}, \quad m\in\{1,\ldots,M\},
  \end{cases}
\end{equation}
with $X_0=x \in \R$ and $Y_{m,0}=y_m\ge 0$ for $m\in\{1,\ldots,M\}$, and parameters $a_m\ge 0$, $\sigma_m>0$, $b_m\in \R$ and $\rho_m\in[-1,1]$. Here, $(W_1,B_1,\ldots,W_M,B_M)$ a $2M$ independent Brownian motions and the processes $r$, $H$ and $(W_1,B_1,\ldots,W_M,B_M)$ are assumed to be independent. This framework embeds several important models in Mathematical Finance:
\begin{itemize}
  \item When $M=2$, $H\equiv0$ and $r_t=r$ is constant, we recover the double Heston model proposed by Christoffersen et al.~\cite{CHJ}.
  \item When $M=1$,  $r_t=r$ is constant and  $H_t = \sum_{k=1}^{N_t}J_k$ is an independent compound Poisson process written through a Poisson process $N$ with constant intensity $\lambda$ and the i.i.d. jumps $\{J_k\}_{k\in\N}$, we obtain either Bates' model~\cite{Bates} when $J_1$ is normally distributed. When $J_1$ is an asymmetric double exponential distribution, we obtain an extension of Kou's model~\cite{Kou}.
  \item When $M=1$, $H_t = \sum_{k=1}^{N_t}J_k$ is a  compound Poisson process with normal jumps and the interest rate $r_t$ follows the Cox-Ingersoll-Ross model, we recover the model proposed by Bakshi, Cao, and Chen \cite{BCC97}.
\end{itemize}
We assume -- which is the case for all the listed models above -- that the processes $r$ and $H$ can be sampled exactly on any discretization grid. We simulate the processes 
$Y_m$ on the fine and coarse grids exactly as in Subsection~\ref{Subsec_implementation}, using either the Ninomiya-Victoir scheme (if $\sigma_m^2\le 4a_m$) or the exact one. We note $\hat{Y}^0_m$ the scheme on the coarse grid and $\hat{Y}^1_m$ on the fine grid. For the log-stock $X$, we use the following scheme for $1\le k\le n$ if $\ell=0$ and for $k\not= \kappa+1$ if $\ell=1$:
\begin{align*}
  \hX^{\ell}_{kh_1} =&\hX^{\ell}_{(k-1)h_1} +\frac{h_1}{2}(r_{(k-1)h_1}+r_{kh_1}) + H_{kh_1}-H_{(k-1)h_1} \\&+\sum_{m=1}^{M}\Bigg(-\frac{\rho_m a_m}{\sigma_m} h_1 +\frac{\rho_m}{\sigma_m}(\hY^\ell_{m,kh_1}-\hY^\ell_{m,(k-1)h_1}) +(\frac{\rho_m}{\sigma_m}b_m-\frac{1}{2}) \frac{h_1}{2}(\hY^\ell_{m,(k-1)h_1}+\hY^\ell_{m,k h_1}) 
  \Bigg)\\&+ \sqrt{\sum_{m=1}^M(1-\rho_m^2)\frac{h_1}{2}(\hY^{\ell}_{m,(k-1)h_1}+\hY^{\ell}_{m,k h_1})}N_{k}
\end{align*}
and, for $\ell=1$ and $1\le k\le n$,
\begin{align*}
  &\hX^{1}_{\kappa h_1+kh_2} = \hX^{1}_{\kappa h_1+(k-1)h_2} +\frac{h_2}{2}(r_{\kappa h_1+(k-1)h_2}+r_{\kappa h_1+kh_2})+ H_{\kappa h_1+kh_2}-H_{\kappa h_1+(k-1)h_2}  +\sum_{m=1}^{M}\Bigg(\\&-\frac{\rho_m a_m}{\sigma_m} h_2 +\frac{\rho_m}{\sigma_m}(\hY^1_{m,\kappa h_1+kh_2}-\hY^1_{m,\kappa h_1+(k-1)h_2}) +(\frac{\rho_m}{\sigma_m}b_m-\frac{1}{2}) \frac{h_2}{2}(\hY^1_{m,\kappa h_1+(k-1)h_2}+\hY^1_{m,\kappa h_1+kh_2}) 
  \Bigg)\\&+ \sqrt{\sum_{m=1}^M(1-\rho_m^2)\frac{h_2}{2}(\hY^{1}_{m,\kappa h_1+(k-1)h_2}+\hY^{1}_{m,\kappa h_1+k h_2})}\tilde{N}_{k}.
\end{align*}
Notice that we use a single Gaussian increment for the $X$ component, taking advantage of the independence between $W_1,\dots,W_M$. We have used the same notation as in Subsection~\ref{Subsec_implementation}, and in particular we can use the same coupling as in~\ref{Subsec_coupling} between $N_{\kappa+1}$ and $(\tilde{N}_k)_{1\le k \le n}$. Simulating on the whole grid enables us to price pathwise options such as Asian options, as explained in Subsection~\ref{Subsec_Pricing}. If only European are considered, we can use instead the one-step method presented in Subsection~\ref{Subsec_coupling}: it requires to use a single normal variable for the log-stock component and gives a lower variance for the correcting term. 
}

\bibliographystyle{abbrv}
\bibliography{Biblio_Heston}

\end{document}